\documentclass[11pt]{article}
\usepackage[dvips]{graphicx}
\usepackage{color}
\usepackage[x11names,rgb]{xcolor}
\usepackage[margin=1cm]{geometry}
\usepackage{fancyvrb,xcolor}
\usepackage{amsfonts}
\usepackage{amsmath, bm}
\usepackage{amssymb}
\usepackage{amsthm}
\usepackage{enumitem}
\usepackage{venndiagram}
\usepackage{booktabs}
\usepackage{multirow}
\usepackage{authblk}
\usepackage[pdftex]{hyperref}
\usepackage{natbib}
\usepackage{fancybox}
\usepackage{hyperref}
\usepackage{authblk}
\usepackage[normalem]{ulem}
\usepackage[font={small},labelfont=bf]{caption}
\usepackage[draft]{changes}
\usepackage{setspace}
\usepackage{hyperref}
\usepackage{parskip}

\hypersetup{colorlinks,            citecolor=blue,            filecolor=blue,            linkcolor=blue,            urlcolor=blue,            }
\renewcommand{\vec}[1]{\mathbf{#1}}

\newtheorem{theorem}{{\bf \sc Theorem}}
\newtheorem{proposition}[theorem]{{\bf \sc Proposition}}
\newtheorem{lemma}[theorem]{{\bf \sc Lemma}}

\newtheorem{corollary}[theorem]{{\bf \sc Corollary}}
\newtheorem{definition}{{\bf \sc Definition}}

\theoremstyle{definition}

\newtheorem{examplex}{Example}
\newenvironment{example}
  {\pushQED{\qed}\examplex}
  {\popQED \endexamplex}
\input{tcilatex}

\begin{document}

\title{Revealing information -- or not -- in a social network of traders\footnote{%
This research was supported by the  
Italian Ministry of Education Progetti di Rilevante Interesse Nazionale (PRIN) grants P20228SXNF and 2022389MRW.
We thank participants at seminars at
Antwerp, Cambridge, Carlos III, IMT Lucca, Siena, Virginia Tech and Wisconsin, as well as participants at the CEPET 2023, $9^{th}$ Conference on Network Science and Economics, SAET 2024, SENA Workshop, and Workshop
on Games on Networks in Singapore.}}

\author[a]{Patrick Allmis}
\author[b,c]{Paolo Pin}
\author[d]{Fernando Vega Redondo}

\affil[a]{University of Cambridge}
\affil[b]{Department of Economics and Statistics,
Universit\`a di Siena, Italy} 
\affil[c]{BIDSA,  Universit\`a  Bocconi, Milan, Italy}
\affil[d]{Department of Economics, The Chinese University of Hong Kong}  
 %

\date{
June 2026}
\maketitle

\setcounter{page}{0}
\maketitle
\thispagestyle{empty}

\abstract{We build upon a simple micro-founded model of asset trading 
proposed by \cite{kyle1985continuous} 
to study under what conditions a trader who is privately informed of the future return of the asset may want to share her information with other traders. Despite what conventional wisdom suggests, 
we show that in the unique equilibrium of the game the informed trader reveals her information with positive probability. 
A consequence of it is that, in contrast with {  the corresponding no-communication benchmark}, the equilibrium price need not be fully revealing of the asset's return, even if traders are risk neutral. This, in turn, has significant implications on the distribution of the social surplus.
While our model initially assumes that inter-agent communication is restricted by an arbitrarily given social network, we also study which such networks arise when links are endogenously formed through traders' prior connection decisions.}

\newpage

\section{Introduction}

There is a lively theoretical literature concerned with the issue of whether, under rational expectations, equilibrium prices may reveal the information that would otherwise lie dispersed across the market participants. Indeed, this question has been addressed within a variety of different theoretical frameworks, the answers being positive or negative depending on a number of key modeling details. When focusing specifically on financial markets, a common implicit assumption is that agents can rely \textit{only} on prices to access (fully or just partially) the information originally held by others. What this implicit assumption has ruled out is the possibility that relevant information might be shared through inter-agent communication before trade is completed. This, however, is hardly something that can be taken for granted in many financial markets where the number of sizable agents is quite small and therefore may enjoy frequent opportunities to ``talk." Couldn't they then have incentives to share payoff-relevant information?  \medskip

Possibly, the intuitive reason why this question has been ignored in the study of financial markets derives from the fact that private information on asset returns is often very valuable if exploited individually. Therefore, it is natural to suppose that no strategic trader in those markets should be inclined to sharing such information with others. In alternative contexts, however, this question has been studied quite exhaustively -- in particular, within the classical framework of Oligopoly Theory. A common approach in this case has been to add to one of the traditional setups (e.g. Cournot or Bertrand) a prior information-sharing stage in which firms can commit to sharing their information with others. This literature will be briefly summarized below. For the moment, it suffices to advance that the aforementioned intuition (i.e.  that no sharing should occur at equilibrium) does apply in some cases but not in others. One of our subsidiary objectives in this paper will be to understand how those results in the Theory of Industrial Organization compare with our own (whose primary focus is on financial markets). \medskip

To address the issue in a transparent manner, we build upon a stylized abstract context proposed by \cite{kyle1985continuous}, where the forces at work are especially clear. It considers a secondary market in which strategic agents trade in an asset/good whose \textit{ex-post} return/value is uncertain. To fix ideas, it can be conceived as a stylized representation of a market for a financial asset (say, government bonds) that is available in some fixed supply.\footnote{The model could also reflect a market for the rights of exploitation of a natural resource of uncertain quality or quantity, or the shares of an IPO for a private firm of uncertain future profitability that turns to being publicly traded.} We depart from the framework in \cite{kyle1985continuous} by allowing strategic communication and by assuming that the demand of the non-strategic sector is common knowledge, so that prices would be fully informative in the absence of communication. Specifically,
before trade takes place among a set of large (strategic) traders, 
one of them 
receives a private signal that informs her of the future return of the asset. 
Then she decides (optimally, by perfectly anticipating subsequent equilibrium behavior) whether to share this signal with her neighbors/friends in the social network of other strategic traders. Once this decision is taken by the informed trader, herself and every other strategic trader simultaneously choose optimally (i.e. also in equilibrium, on the basis of the information each of them holds) the magnitude and sign of their trading position. Given those trades and the aggregate trade induced by large number of infinitesimal non-strategic agents (whose aggregate behavior is captured by a downward-sloping demand function), the market price adjusts to clear the whole market.\medskip


We start our analysis by showing that, at equilibrium, there is \textit{always some positive probability} for the  privately informed agent to reveal the return of the asset to her network neighbors. Intuitively, this occurs when the informed agent has someone to reveal to and the news about this return that the informed agent receives is neither very good nor very bad. For, in either of these two cases, the gain that she can obtain from exploiting privately such ``extreme" information more than offsets the benefit of limiting competition by other strategic traders when the information is shared. We also find that, because of the additional ``decoding uncertainty" induced by the conditional behavior of the party holding the information, {  whenever the informed trader has at least one out--neighbor,} there is always positive probability that the equilibrium price fails to reveal that information.  Thus, returning again to the issue of price revelation, we find that even in the simple Kyle--inspired context we explore where, in the absence of the possibility of revelation, the price would be fully informative, it is precisely such a possibility that renders it always possible that it be non-revealing.\medskip

Naturally, the possibility that information can be strategically revealed must be expected to have distributional implications in our context. For, in effect, the set of all agents (including the non-strategic traders and the original owners of the asset) are involved in a game that is \textit{ex post} zero-sum. We find that this is indeed the case since, whenever the information is revealed, the equilibrium price always falls. Hence, as this happens, the original owners -- and only these -- experience a payoff loss. \medskip 

All of the previous discussion pertains to the case where the social network that channels communication is exogenously given. What would happen if the network is instead endogenously determined through agents' own choices? To address this question, we conclude the paper by studying a network-formation game where, in a first stage, agents simultaneously announce the set of other players they would like to be connected to. When two agents reciprocally announce each other, the link between them is formed and they can then communicate in the ensuing trading game, which proceeds as described before. A sharp result obtains if we apply an extension of the Nash Equilibrium concept, Bilateral Equilibrium, that allows for bilateral deviations (which must obviously be contemplated, if we make the natural assumption that any new link requires bilateral agreement of the agents involved). We show that,  under these conditions, the network can be seen as formed through a dynamic process that yields uniquely a configuration displaying the following simple structure: agents are partitioned into cliques (i.e.~completely connected sub-networks), with traders being selected in a decreasing order in terms of their individual probability of being informed. In the end, this provides a self-enforcing (i.e.~equilibrium) balance between the two opposing forces that determine whether a clique wants to add the best trader yet available. On the one hand, it has to assess the \textit{marginal benefit} of adding an additional source of information. On the other hand, it has to account for the possible \textit{marginal loss} entailed from sharing its full information potential with one other competitor.

\subsection*{Related Literature}

We close this introduction with a brief summary of related literature. First, we refer to the literature that has studied the issue of price-induced revelation in market setups.  A seminal contribution was provided by \cite{grossman1980impossibility}. They showed that, under rational expectations, a competitive market will have the price reveal all relevant information available to market participants. This raised the well-known paradox that, under these conditions, no agent would ever want to incur any cost in acquiring information -- thus, in the end no information would be obtained. To address this paradox, the literature reformulated the question in a context where agents behave strategically, understanding the impact of their decisions on the market price. Early influential papers are \cite{Dubey1987revelation} and \cite{kyle1989schizophrenia}%
, while a good recent survey on the topic can be found in \cite{rostek2020equilibrium}. As mentioned, this literature shows that the ability of equilibrium prices to be a sufficient statistic of the decentralized private information held by the agents depends on detailed characteristics of the model, e.g.~the role played by aggregate or/and individual shocks, the (a)symmetry of agents' signals and their (un)correlation, or whether they may submit demand/supply schedules or only price-independent trades. In our model, which is designed to focus on the analysis of information sharing, price revelation would always materialize if agents did not have the option of sharing their information. Thus, if such revelation fails to happen, it must be because of that possibility.\medskip

The extension of oligopoly theory that, as outlined before, introduced into this classical field the possibility of information sharing has given rise to a large body of literature.  Starting with the papers of \cite{NovSonn1982}, \cite{vives1984duopoly}, and \cite{Gal-Or1985}, the ensuing research has explored a wide range of different scenarios (in terms of the substitutability/complementariy of the goods produced, the nature and object of the uncertainty, the type and relationship of the signals received, etc.). A systematic analysis and comparison of these different scenarios is provided by \cite{Raith1996}. The specific conclusions vary widely, depending on the scenario being considered. But, abstracting from intricate details, what  arises from this analysis is that there is a large range of circumstances for which firms have a strong incentive to share information and therefore, at equilibrium, the firms will choose to do it in the first stage of the extended oligopoly game. Interestingly, one exception where this does \textit{not} happen is in a context that, as we shall see, formally resembles our model, i.e. the standard Cournot model where firms produce a homogeneous good and compete in quantities, with the uncertainty pertaining only to a common-value payoff component of (e.g.~a parameter shifting market demand). The reasons for this sharp contrast with our model (which, as explained, always attributes positive probability to revelation) will be explained at some length once our model has been discussed in detail.\footnote{ Another interesting contrast between our model and classical oligopoly is the following. While in oligopoly the welfare impact of information sharing can be positive or negative (both if we consider total or consumer surplus -- see \citealt{vives1984duopoly}), in our case the total surplus always remains constant. The reason is that, \textit{ex post},  our game is zero-sum, hence the original owners always end up weakly worse off when information is shared (and consequently the price falls) while all the other traders become weakly better off.} \medskip


The paper proceeds as follows. Section \ref{sec:model} sets up the model. Section \ref{sec:analysis} solves the model. Section \ref{S-Info_eff} discusses the informational efficiency of the market at equilibrium. Section \ref{sec:endogenous} studies the networks formation game. Section \ref{sec:conclusion} concludes.

\section{The Model}\label{sec:model}
We analyze a multiperiod game with private information, endogenous information sharing, and trading.
Proceeding backwards in the decision structure of the game, we start by specifying how trade is conducted and profits earned in the last trading stage and then turn to describing how information is received and possibly revealed in the initial information-acquisition (and revelation) stage.  

\subsection{The trading stage}\label{SS-T_stage}
The trading stage is largely inspired in a model proposed by \cite{kyle1985continuous}. It is useful to decompose it into the following two periods.\medskip

In period 1, an amount $Q \geq 0$ of a certain good is made available (i.e. supplied inelastically). To fix ideas, we shall think of this good as an asset that delivers a (possibly negative) random return per--unit $\tilde{V}$ in period $2$. As will become clear later, the equilibrium price is non-negative. In the market for this asset, there are two types of agents active:\medskip
\begin{itemize}[itemsep=4 pt, topsep=0 pt]
    \item a \emph{large} set of non-strategic agents whose aggregate demand is given by $H-K P$, where $H \geq 0 $ is the (possibly random) vertical intercept of their demand function, $K>0$ is a fixed parameter that captures their sensitivity to price changes, and $P$ is the price of the good (which will be determined through a market-clearing condition specified below);
    \item a finite set $\mathbf{N} = \{1,2,...,N\}$ of strategic agents (we simply refer to them as ``traders"), each $i \in \mathbf{N}$ submitting a demand (fixed order) for a \textit{net} quantity  $x_i$. These magnitudes can be negative (in principle, also for non--strategic traders if the price $P$ is high), which is interpreted  as reflecting a situation where traders go short.
\end{itemize}  \medskip

In period 2, the return $V$ of the asset is realized, (as it has already happened for $Q$ and $H$ in period 1), which leads, for  each trader $i \in \mathbf{N}$, to profits
\[
\pi_i = x_i (V - P).
\]
that are proportional to the net position achieved by each $i$ through the trade conducted in period $1$. As advanced, the price is set so as to clear the market:
\[
Q = \sum_{i \in \mathbf{N}} x_i +H -KP  ,
\]
and therefore we have:
\begin{equation} \label{eq:price}
P = \frac{X+H-Q}{K},
\end{equation}
where $X \equiv \sum_{i \in \mathbf{N}}x_i$ is the aggregate net quantity of the asset demanded by the (strategic) traders.

\subsection{Information-Acquisition Stage}\label{SS-IA_stage}
Let us assume that, before the game starts, all agents know the true distribution of the random variable $\tilde{V}$, as well as the deterministic variables $H$, and $Q$ and this is common knowledge. We shall see that the whole payoff-relevant information is fully encapsulated by the random variable $\tilde{r} \equiv K\tilde{V} - H + Q$.\footnote{All results presented in this paper would survive if $\tilde{H}$ and $\tilde{Q}$ were random variables as well (in particular, taking into account any possible correlations that might exist among its constituent variables: $\tilde{V}$, $\tilde{H}$, and $\tilde{Q}$). We make this restriction for exposition, since the payoff relevant information players can share with out-neighbors always concerns the asset's ex-post return. This is relevant only in Proposition \ref{prop_price}.} Thus our analysis will center on it as the only piece of information that agents should care about. 
{  We allow \(\tilde V\), and hence \(\tilde r\), to take  negative real values. Thus,}
we assume that $\tilde{r}$ is distributed according to a probability density function $f(\cdot)$ with convex support $\Omega$ that, for simplicity, we shall identify with the whole real line $\mathbb{R}$. This distribution is assumed to have no points of positive mass (hence its cumulative distribution function $F(r)$ is monotonic and differentiable) and it has a strictly positive mean that we denote by $\bar{r} \equiv E(r)= \int_{\Omega} r dF(r) >0$. \medskip

Suppose that this first stage starts by at most one single agent being informed by Nature of the realization of $\tilde{r}$, which is indicated by $r$. The \textit{ex-ante} probabilities that any given agent $i$ receives that perfectly revealing signal is denoted by $p_i$, and we allow that $\sum_{i \in \mathbf{N}}p_i < 1$  so that it is possible that no agent receives any signal at all.\medskip

Agents are taken to be connected through a directed social network of communication, $\Gamma = \{\mathbf{N}, L\}$, where $L \subset \mathbf{N}\times\mathbf{N}$ and a pair $(i,j) \in L$ means that $i$ ``talks" to $j$ and therefore $i$ can choose to share with $j$ any information she might obtain. The subset of all those other agents who can be informed by $i$ are called her out-neighbors and denoted by $\hat{N}_i$. The cardinality of this set, called the out-degree of $i$ is denoted by $\ell_i$. Reciprocally, we denote by $\check{N}_i$ the set of in-neighbors of $i$, who are the subset of agents that ``talk" to $i$. { 
For tractability}, we assume that if an agent $i$ receives the information from Nature, she has only two revelation decisions: either share it with \textit{all} her out-neighbors in $\hat{N}_i$ or with none at all. { However, in Section \ref{SS-full-game} we explain that our main conclusions would extend to the equilibria of a generalized model  where any informed $i$ may choose to reveal her information to any subset of agents in $\hat{N}_i$.}\footnote{{A different extension worth considering would be to allow for the possibility that agents who receive the information \textit{indirectly} from others may further relay it to their (non-informed) out-neighbors. To discuss this issue, in Section \ref{SS-full-game} we also outline a variation of our theoretical framework that admits this possibility and can be studied in a way analogous to our model.} Note, however, that in many real-world contexts (e.g. in some financial markets) speed is an issue because fresh information tends to become obsolete very fast. An additional support for this assumption will be provided in Section \ref{sec:endogenous}, where we study a network-formation model that endogenizes the communication network. There we will show that, under suitable assumptions, agents organize themselves in cliques (completely connected and disjoint), thus making multi-step relaying of information irrelevant.}\medskip

A graphical description of the setup is offered in Figure \ref{fig:network}.
Its left-hand network represents the ex--ante pattern of information transmission.
Once a trader is informed the effective situation simplifies, as shown in the right-hand network.
The red node in it is the trader who has received the signal and is originally informed, while the green nodes represent traders who can be subsequently informed by the red node. Finally, the remaining traders, represented by the blue nodes, cannot receive the information in any case because the red node cannot access them.\medskip

\begin{figure}
\begin{center}
\includegraphics[height=4cm]{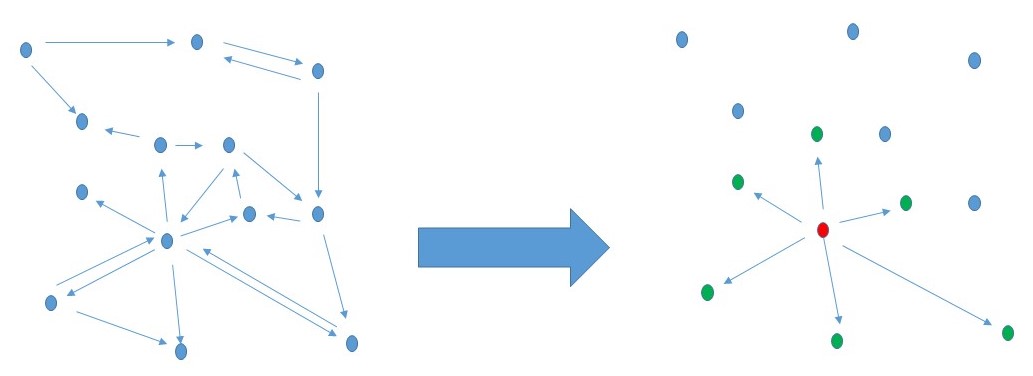}
\caption{A graphical illustration of the informational asymmetry that arises (from left to right) when trader $i$ receives a perfectly revealing signal from Nature. Once this trader (depicted in red) is so informed, she can decide whether or not to reveal the signal to those other traders (in green) to whom she is connected. All others (in blue) cannot receive any information.}
\label{fig:network}
\end{center}
\end{figure}

Finally, note that once an agent $i$ receives the signal and becomes informed, another modeling question that needs to be addressed is \textit{who learns about that event}, i.e. what agents get to know that $i$ has been informed, even if they do not (or cannot) learn the content of the information. 
Our leading assumption in this respect will be that, whenever Nature informs trader $i$, 
all other traders in $N \setminus \{i\}$ observe that $i$ has been informed, 
although only $i$ (and, upon revelation, her out-neighbors in $\hat N_i$) observe 
the realization of $r$. 
An alternative possibility would be to consider a \emph{neighbors-know scenario}, 
where only the out-neighbors in $\hat N_i$ observe that $i$ has been informed. 
However, as we show after Proposition~\ref{prop_characterization}, 
from the viewpoint of the equilibrium allocation and prices, 
both informational specifications are equivalent.

\subsection{The Full Game}\label{SS-Game}
Having described in Subsections \ref{SS-IA_stage} and \ref{SS-T_stage} the information-acquisition and trading stages, an integration of them in this order defines the full Bayesian game played by the (strategic) traders where not only shall we make the assumption that the \textit{ex-ante} distributions are common knowledge but that the same applies to all other details of the interaction specified before -- in particular the social network of communication $\Gamma =\{\mathbf{N},L\}$. To sum things up, it is useful to describe its extensive form as follows.  \medskip
\begin{enumerate}
    \item In the first stage, the following steps are taken in sequence.
    \begin{enumerate}
        \item Nature observes the prevailing value of $r \in \Omega$ and informs at most  a single agent $i$ about it,  according to the probabilities  $\{ p_1, \dots , p_N,\hspace{1 pt}1-\sum_{i \in \mathbf{N}}p_i \}$.
        \item The informed agent $i$ decides whether or not to reveal $r$ to (all or none of) her out-neighbors in the set $\hat{N}_i$. 
    \end{enumerate}
    \item In the second stage, the following actions are chosen simultaneously:
    \begin{enumerate}
        \item The informed agent $i$ chooses $x_i$, depending on the value of $r$ and on whether she revealed it or not in item (b) of stage 1.
        \item Each $j \in \hat{N}_i$ (an out-neighbor of $i$) chooses $x_j$ depending on whether $i$ revealed the information she obtained in stage 1 (in which case $x_j$ may depend on $r$) or not.
        \item Each $k \in \mathbf{N}\setminus \left(\hat{N}_i \cup \{i\}\right)$ chooses $x_k$ 
    \item Each agent $j \in \mathbf{N}$ obtains a payoff 
\[\pi_j =\frac{1}{K} \left(r-\sum_{k \in \mathbf{N}} x_k\right) x_j=\frac{1}{K}\left(KV-H+Q-X\right)x_j = (V-P)x_j.\]
\end{enumerate}
\end{enumerate}\medskip

Building upon the above description of the extensive form of the game, one can specify the general format of a typical strategy for any trader $i \in \mathbf{N}$. Such a strategy $s_i = \left[ g_i(\cdot), x_i(\cdot),\bar{x}_i(\cdot),y_i(\cdot),\bar{y}_i(\cdot),z_i(\cdot) \right]$ includes the following components:\footnote{By way of mnemonic help, the notation "$x$" and "$y$" is respectively used to distinguish 
the trades chosen by the agent who has been informed of $r$ by Nature and those of her out-neighbors. 
A bar on top of these variables signifies that such information has not been revealed, while the absence 
of the bar indicates that it has. Throughout the paper, when trader $i$ is directly informed by Nature, 
we refer to $i$ as the \emph{informed} trader, to $j \in \hat N_i$ as \emph{possibly informed} traders, 
and to $k \in N \setminus (\{i\} \cup \hat N_i)$ as \emph{always uninformed} traders (in the sense that 
they never observe $r$ directly). As will follow from Proposition \ref{prop_characterization} below, when revelation occurs and all 
strategic traders observe $r$, equilibrium behavior becomes symmetric and all agents trade $y_j(i)=\tfrac{r}{N+1}$. 
This coincides with the trade that $z$-traders would choose if they held the correct posterior belief about $r$.}\medskip
\begin{itemize}[itemsep=3 pt, topsep=0 pt]
\item a function $g_i: \Omega \rightarrow \{0,1\}$ that, if $i$ is informed of $r$ by Nature, determines those values for which she reveals it to her out-neighbors ($g_i(r)=1$) or conceals it from them ($g_i(r)=0$);
\item a function $x_i: \Omega \rightarrow \mathbb{R}$ such that if Nature has informed $i$ about $r$ and she has revealed it to her out-neighbors in $\hat{N}_i$, then $x_i(r)$ determines her $r$-contingent trade;
\item a function $\bar{x}_i: \Omega \rightarrow \mathbb{R}$ such that if Nature has informed $i$ about $r$ and she has \textbf{\textit{not}} revealed it to her out-neighbors in $\hat{N}_i$, then $\bar{x}_i(r)$ determines her $r$-contingent trade;
\item a function $y_j: \Omega \times \check{N}_j \rightarrow \mathbb{R}$ such that if  $i$ is an in-neighbor of $j$ who has revealed the value of $r$ to her, then  $y_j(r,i)$ determines $j$'s ($r,i$)-contingent trade;
\item a function $\bar{y}_j: \check{N}_j \rightarrow \mathbb{R}$ such that, if  $i$ is an in-neighbor of $j$ who has \textbf{\textit{not}} revealed the value of $r$ to her, then $\bar{y}_j(i)$ determines $j$'s $i$-contingent trade;
\item a trade $z_k: \big( (N \setminus \{k\}) \cup \{\emptyset\} \big) \rightarrow \mathbb{R}$ chosen by trader $k$ in either of the following cases: 
(i) no trader is informed by Nature; or 
(ii) some trader $i$ is directly informed and 
$k \in N \setminus \big(\{i\} \cup \hat N_i\big)$.
In either case, trader $k$ does not observe the realization of $r$, 
although she may observe the identity of the directly informed trader $i$.
\end{itemize}

\section{Analysis of the Model}\label{sec:analysis}
This section conducts the basic analysis of the model. We start in Subsection \ref{SS-Equil_trading} by characterizing the equilibrium in the simultaneous trading subgames arising from any possible outcome of the initial information-acquisition stage. Subsection \ref{SS-PBE} specifies the conditions defining a Perfect Bayesian equilibrium. Finally, in Subsection \ref{SS-full-game} we extend the equilibrium analysis to the full game that integrates the trading and information-acquisition stages.
\subsection{Equilibrium behavior in trading subgames}\label{SS-Equil_trading}
Proceeding backwards in the analysis of the full game, let us consider first a typical situation (i.e. subgame) reached right after the agents have completed the information-acquisition stage. At that point, they must update their (initially common) beliefs on the value of the prevailing $r$ that, as explained, is the all-encompassing market fundamental. In general, such updated beliefs will end up being heterogeneous across agents in ways that depend on who was informed by Nature in the first stage of the game and the architecture of the communication network $\Gamma$. For each $i\ \in \mathbf{N}$, let us denote by $\tilde{r}_i$ the random variable on $\Omega$ that reflects those induced beliefs. It is on the basis of them that agents choose (simultaneously) their trades in the corresponding subgame being reached. \medskip

Traders are profit-maximizing agents and their key concern is to predict the eventual price for the good, as determined by the random variable:\footnote{Throughout the paper, a tilde over a variable denotes the corresponding random variable when the fundamental value $r$ is not known. The corresponding probability distribution to use in that case will either be specified or clear from the context.}
\begin{equation} \label{eq:price_rand}
\tilde{P} = \frac{\tilde{X}+H-Q}{K}.
\end{equation}
which is simply a random counterpart of the market-clearing condition (\ref{eq:price}). More precisely, every trader $i$ uses all the information derived from the first stage -- in particular, whether  someone she knows (i.e. a neighbor) or herself has been informed and has (or has not) revealed -- to form beliefs on the value of $r$ and on the mapping $X_{-i}\equiv \sum_{j \neq i}x_j: \Omega \rightarrow \mathbb{R}$ that associates to every possible value of $r \in \Omega$ an aggregate demand of all other traders. Naturally, this latter random variable should depend on some conjecture about the strategies pursued by others and also be consistent with her beliefs on $r$ . However, for the moment, let us abstract from these matters (which are dealt with in Subsection \ref{SS-PBE}) and simply denote the induced random variable by $\tilde{X}_{-i}$ and by $\mathbb{E}_i(\cdot)$ the general expectation conditional on the information trader $i$ holds at the trading stage. Then, trader $i$ solves the following optimization problem:
\begin{equation}\label{BestResponse}
    \max_{x_i \in \mathbb{R}} \mathbb{E}_i (\pi_i) = x_i \cdot \mathbb{E}_i \left( \tilde{V} - \frac{x_i + \tilde{X}_{-i} + H - Q}{K}  \right) = \frac{1}{K} \left[ x_i \ \mathbb{E}_i \left(\tilde{r} - \tilde{X}_{-i} \right) - x_i^2\right], 
\end{equation}
and the optimal trade $x^*_i$ is given by: 
\begin{equation}
\label{eq:opt_x}
x_i^*  = \frac{1}{2} \Big( \mathbb{E}_i (\tilde{r}) - \mathbb{E}_i (\tilde{X}_{-i})  \Big).
\end{equation}

Note that, as anticipated in the Introduction, the best-response function given by \eqref{eq:opt_x} exhibits the same formal structure induced by Cournot competition with a linear demand curve and a random intercept, as in the model of \cite{vives1984duopoly}. The reason why,  
despite such a formal parallelism, revelation always happens to be a positive-probability event in our model, while it never arises at equilibrium in the model studied by Vives, lies in the different type of  informational externalities that revelation creates in each case. In Vives' model, if demand information were to be revealed by a Cournot oligopolist, this would reduce market uncertainty overall, which in turn would stimulate aggregate production and deliver a positive externality on buyers (thus, a negative one on firms). Instead, in our model, information revelation across firms reduces overall uncertainty by the (net) buying side of the market, i.e. the strategic traders. Thus, in some cases (specifically, when the information initially obtained is not too extreme), it is worth revealing it to other traders in order to mitigate competition and keep the market price low. In those cases, therefore, the negative externality induced by revelation of information is experienced by the non-strategic seller.

\subsection{Perfect Bayesian Equilibrium of the full game}\label{SS-PBE}
To complete the analysis of the game, we need to specify how traders shape their beliefs on $r$ that underlie equilibrium trading behavior in its second stage, after the first information-acquisition stage has been completed. These beliefs are formed in one of two different ways: 
\begin{itemize}
    \item[(a)] \textit{directly through revelation}, either by Nature or by some in-neighbor in the social network who has been informed by Nature herself;
    \item[(b)] \textit{indirectly through Bayesian updating}, after observing that either (i) some in-neighbor has been informed by Nature but has \textit{not revealed} this information or (ii) no in-neighbor (nor herself) has been informed by Nature.
\end{itemize} \medskip

As usual, to understand the situation prevailing at any point along the game,  we must specify the belief that the deciding trader $i \in \mathbf{N}$ would hold  on the value of $r$ if such a point were reached. In general, those beliefs depend not only on the information set $h$ at which the trader finds herself but also on the \textbf{\textit{conjectured}} strategy profile $\mathbf{s}=(s_1,s_2,...,s_N)= \left[ g_i(\cdot),x_i(\cdot),\bar{x}_i(\cdot),y_i(\cdot),\bar{y}_i(\cdot),z_i \right]_{i\in\mathbf{N}}$  governs the behavior of all traders. In view of these (fully standard) considerations, we denote those beliefs by $\tilde{r}_i(h \vert \mathbf{s}, r)$, conceiving it as a real random variable. Then, if we require that all those beliefs result from a proper application of Bayesian updating and correct strategy conjectures, and that traders behave optimally given those beliefs in all of their information sets, we arrive at a minimalist notion of Perfect Bayesian Equilibria (PBE) that is sufficient in our context.\footnote{We do not contemplate any of the conditions often demanded from off-equilibrium beliefs because all information sets in the second stage of the game are visited with positive probability in any Nash equilibrium. Thus the requirement of perfection (optimal responses to the uniquely induced set of beliefs) is enough to define the equilibrium.} \medskip

To describe formally the PBE notion, the following piece of notation is useful. First we partition the information sets of each agent $i$, denoted by $\mathcal{H}_i$, into (a) the subset $\mathcal{H}^1_i$ that includes the singleton information sets $\left[h^1_i\hspace{-2 pt}=\hspace{-2 pt}\{r\}\right]_{r \in \Omega}$ that arise in the first stage when Nature informs $i$ of the value of $r$, and (b) the subset $\mathcal{H}^2_i$ that comprises  the remaining information sets, all in the second stage. Second, denote by $x_i(h_i \rvert \mathbf{s})$ the trade chosen by agent $i$ at information set $h_i \in \mathcal{H}_i$ when the strategy profile across all traders is $\mathbf{s}$. And third, with a slight abuse of notation, we have 
$\pi_i: S_1 \times S_2 \times \cdots \times S_N \rightarrow \mathbb{R}$   
stand for the payoff function of our game in strategic form, sometimes writing the strategy $s_i$ of an agent $i$ as $(s^1_i,\mathbf{s}^2_i)$, its first and second components corresponding to the choices made in the first and second stages of the game, respectively.\medskip

\begin{definition}
    A Perfect Bayesian Equilibrium (PBE) of our full game is defined by a strategy profile  
\begin{equation}\label{PBE_strat}
     \mathbf{s^*}=\left[ (g^*_i(r))_{r \in \Omega}\hspace{2 pt},\hspace{1 pt}(x^*_i(r))_{r \in \Omega}\hspace{2 pt},\hspace{1 pt}(\bar{x}^*_i(r))_{r \in \Omega}\hspace{2 pt},\hspace{1 pt} \biggl(y^*_i(r,j)\biggr)_{\substack{
r \in \Omega \\
j \in \check{N}_i
}}\hspace{1 pt},\hspace{1 pt} (\bar{y}^*_i(j))_{j \in \check{N}_i}\hspace{1 pt},\hspace{1 pt}z^*_i \right]_{i \in \mathbf{N}} 
\end{equation}
and corresponding trader beliefs $\big\{\left[ \tilde{r}^*_{i}(h_i\vert\mathbf{s^*})\right]_
{h_i \in \mathcal{H}_i}\big\}_{i \in \mathbf{N}}$ such that, for any given $r \in \Omega$ and each $i \in \mathbf{N}$, the following conditions hold:
\begin{itemize}[itemsep=5 pt, topsep=4pt]
    \item[(a)]  At every information set $h_i \in \mathcal{H}_i$, the updated beliefs $\tilde{r}^*_i(h_i\vert \mathbf{s}^*)$ are derived using Bayes rule.
    \item[(b)] At every information set $h^1_i=\{r\} \in \mathcal{H}^1_i$, $s^{1*}_i(h^1_i)=g^*_i(r)$ is optimal, so that   \\
    \hspace*{50 pt} $\pi_i(\mathbf{s^*}) \equiv \pi_i\big((s^{1*}_i,\mathbf{s}^{2*}_i),\mathbf{s}^{*}_{-i}\big) \geq \pi_i\big((s^{1}_i,\mathbf{s}^{2*}_i ),\mathbf{s}^{*}_{-i}\big)$ for $s^{1}_i \neq s^{1*}_i$.
\item[(c)] At every information set $h_i \in \mathcal{H}_i^2$, the trade
$x_i^*(h_i\mid s^*)=\tfrac12\Big(\mathbb E_i\big[\tilde r\,\big|\,h_i\big]
-\mathbb E_i\big[\tilde X_{-i}\,\big|\,h_i\big]\Big)$
is a best response given $i$'s belief and the profile $s^*$; that is (cf.~\eqref{eq:opt_x}), under $i$'s posterior
$\tilde r_i^*(h_i\mid s^*)$ over the state, as pinned down by Bayes' rule in part~(a), and
\[
\tilde X_{-i}=\sum_{j\neq i} x_j\big(h_j(\tilde r,s^*)\,\big|\,s^*\big)
\]
is the aggregate trade of the remaining agents induced by the realized state $\tilde r$ and the profile $s^*$, each trader $j$ acting at \emph{her own} information set $h_j(\tilde r,s^*)$.
\end{itemize}
\end{definition}

Note, from $i$'s standpoint $\tilde r$ is distributed according to her posterior at $h_i$, and this is the distribution under which $\mathbb E[\tilde X_{-i}\mid h_i]$ is computed. 

\subsection{Equilibrium analysis of the full game}\label{SS-full-game}
    
To carry out our study of the full game, some initial observations are in order. First note that, for any given updated beliefs generated in the first stage, the subgames played in the second (trading) stage always display a unique equilibrium. This implies that, if an agent $i$ is informed by Nature and this agent reveals the information to her out-neighbors, both $i$ and all agents $j \in \hat{N}_i$ must take the same action at equilibrium. That is, the trading choices of all of them must satisfy:\footnote{Henceforth, for the sake of notational lightness, we dispense with attaching the star label to equilibrium magnitudes when no confusion is likely to arise.}

\begin{equation}\label{eq:opt_y_reveal}
    y_j (r,i) = x_i (r), \hspace{10 pt} \forall j \in \hat{N}_i.
\end{equation}

To compute such common trade, denote by $Z_i \equiv \sum_{k \notin (\hat{N}_i \cup \{i\})} z_k$ the aggregate trade chosen by all those agents who are neither $i$ nor any of $i$'s out-neighbors. Then, in view of  \eqref{eq:opt_x}, we can write:
\begin{eqnarray}
x_i (r) & = & \frac{1}{2} \left(  r- \ell_i x_i (r) - Z_i \right)= \frac{r-Z_i}{2 + \ell_i},   
\label{eq:opt_x_reveal}
\end{eqnarray}
where recall that $\ell_i$ denotes the out-degree of $i$, i.e. the cardinality of the set $\hat{N}_i$. Instead, when the originally informed agent $i$ chooses not to reveal her information, we can again rely on \eqref{eq:opt_x} to compute her trade, which we have denoted by $\bar{x}_i(r)$, as follows: 
\begin{eqnarray}
\bar{x}_i (r) & = & \frac{1}{2} \left(  r- Z_i - \sum_{j \in \hat{N}_i } \bar{y}_j(i) \right).
\label{eq:opt_x_hide}
\end{eqnarray}
where recall that $\bar{y}_j$ denotes the common action chosen by $i$'s out-neighbors $j \in \hat{N}_i$ when they are kept uninformed. \medskip

Building upon (\ref{eq:opt_x_reveal})-(\ref{eq:opt_x_hide}), we arrive at the following preliminary result.\medskip

\begin{lemma}
\label{lemma_ab}
Suppose agent $i \in \textbf{N}$ is informed by Nature of the prevailing value of $r$. Then, if $\ell_i >0$, there is some \emph{(}nonempty\emph{)} real interval $\left( a_i, b_i  \right)$ such that, at any PBE \ $\mathbf{s}$ of the full game, agent $i$ reveals $r$ to her out-neighbors \emph{(}i.e. $g_i (r) = 1$\emph{)}, if and only if $r \in \left( a_i, b_i  \right) $,\footnote{
Here we disregard the cases where $r$ is equal to $a_i$ or $b_i$, since they correspond to zero-probability events.} where
\begin{eqnarray}
a_i & \equiv & \min \left\{ \frac{2+\ell_i}{4+\ell_i} \left( \sum_{j \in \hat{N}_i } \bar{y}_j(i) \right) + Z_i , \frac{2+\ell_i}{\ell_i} \left( \sum_{j \in \hat{N}_i } \bar{y}_j(i) \right) + Z_i  \right\} \ \ 
\label{eq_ai}
\\
b_i & \equiv & \max \left\{ \frac{2+\ell_i}{4+\ell_i} \left( \sum_{j \in \hat{N}_i } \bar{y}_j(i) \right) + Z_i , \frac{2+\ell_i}{\ell_i} \left( \sum_{j \in \hat{N}_i } \bar{y}_j(i) \right) + Z_i  \right\} \ \ .
\label{eq_bi}
\end{eqnarray}
\end{lemma}\medskip

The above lemma shows that, in every PBE of the game, the revelation decision taken by any  agent $i$ who is informed of  $r$ by Nature is a quite natural one: she reveals that information only if the value of $r$ lies within a certain ``intermediate" interval. Intuitively, this reflects the fact that, if the value of $r$ is too low (less than $a_i$) or too high (above $b_i$), the incentives for agent $i$ to keep that information private are too strong. For she can then take advantage of her private information to, say, unilaterally go short or long if $r < a_i$ or $r>b_i$, respectively.  Instead, if the value of $r$ is moderate (i.e. $a_i \leq r \leq b_i$), $i$'s primary concern is the mitigation of inter-trader competition. And this is best attained by moving into a symmetric informational situation  with her out-neighbors, sharing the news that there are no substantial gains to be reaped by playing aggressively. 
\medskip

Lemma \ref{lemma_ab} provides an important first step in the analysis of the model but does not yet establish the existence of an equilibrium with revelation. In essence, it provides a \emph{necessary} condition for it: if such an equilibrium exists, revelation must take place only within the specified interval for $r$.   But, in principle, there is no guarantee that such an interval -- whose endpoints are determined from (\ref{eq_ai})-(\ref{eq_bi}) by the endogenous trades $(\bar{y}_j)_{j \in \hat{N}_i}$ and $Z_i \ (\equiv \sum_{k \notin (\hat{N}_i \cup \{i\})} z_k$) -- can be made consistent with all other equilibrium conditions. This consistency requires that those trades, which are chosen by \textit{all other} strategic traders,  
represent optimal choices for them when they access no information. 
{ The existence, uniqueness, and characterization of the solution to the resulting fixed-point problem are established in two steps. We begin with the following lemma.}\footnote{%
Because $Z_i$ is determined by the aggregate trade of agents who do not and cannot observe $r$, 
we can equivalently represent the game with a timing in which $Z_i$ is determined prior to the trading stage. 
This allows us to analyze the subgame conditional on a given value of $Z_i$.}
\medskip

\begin{lemma}\label{prop_existence}
    Consider the subgame in which Nature informs some agent $i$ of the realization of $r$.
{     Fix an aggregate trade $Z_i<\bar{r}$ by agents who neither observe $r$ nor receive it through revelation,
as taken as given by agent $i$ and her out-neighbors in this subgame.}
Then the Bayesian best-response strategies in this subgame are {  uniquely} determined. 
Moreover, the informed agent $i$ reveals the information to her out-neighbors with positive probability whenever $\ell_i>0$.
\end{lemma}\medskip


{ Taken together with Lemma \ref{lemma_ab}, the previous lemma establishes the following two conclusions.} The first one is that { any PBE must involve some degree of separation, if there is an informed trader and she has some out--neighbors.} The second one is that { any} PBE must also induce revelation with the following positive probability:

\begin{equation}\label{rev_prob}
\phi_i \equiv \int_{a_i}^{b_i} dF(r) = F(b_i) - F(a_i)
\end{equation}
for any trader $i$ with a nonempty set of out-neighbors in $\hat{N}_i$.\footnote{%
In defining $\phi_i$ and the following endogenous variables, we express dependency on the informed trader $i$ but not on $Z_i$, as would follow from the assumptions of Lemma \ref{prop_existence}. This is without loss of generality because we prove later in Proposition \ref{prop_characterization} that the PBE is unique.\label{note_noZi}} That is, 
{ any} PBE displays as well some partial ``separation" where some signals are revealed (when $r \in (a_i,b_i)$) while others are concealed (those outside that interval).\medskip

Next, we proceed to provide a useful characterization whereby such equilibrium is \textit{explicitly} associated to its corresponding profile $\vec{\phi}=\left(\phi_i \right)_{i \in \mathbf{N}}$ of revelation probabilities. To this end, we rely on the following notation:
\begin{equation*}
\mu_i \equiv \int_{a_i}^{b_i} r \: dF(r), 
\end{equation*}
which allows us to write the conditional expected values of $r$ upon either revelation or not as follows:
\begin{equation*}
    \mathbb{E}(\tilde{r} \vert r \in (a_i,b_i))=\frac{\mu_i}{\phi_i}
\end{equation*}
\begin{equation*}
    \mathbb{E}(\tilde{r} \vert r \notin (a_i,b_i))=\frac{\bar{r}-\mu_i}{1-\phi_i}.
\end{equation*}
and its unconditional (\textit{ex-ante}) expected value as
\begin{equation*}
    \bar{r} \equiv \int_\Omega r \:dF(r) =  \phi_i\frac{\mu_i}{\phi_i} +(1-\phi_i) \frac{\bar{r}-\mu_i}{1-\phi_i}.
\end{equation*}
Building upon these considerations we now { prove uniqueness and} characterize the PBE of our game. 
{ The following result establishes that the fixed-point problem induced by Lemma \ref{lemma_ab} admits a unique solution and therefore the full game possesses a unique Perfect Bayesian Equilibrium.}
\medskip

\begin{proposition}[equilibrium uniqueness and characterization]
\label{prop_characterization}
{ The full game admits a unique PBE. 
If agent $i \in \mathbf{N}$ is informed by Nature of the prevailing value of r and 
$\left( a_i, b_i  \right)$ denotes her revelation interval (interpreted as empty if $\ell_i=0$), 
then the equilibrium strategies are given by the following expressions:}
\vspace{6 pt}
\begin{enumerate}[itemsep = 10 pt, topsep= 5 pt, partopsep = 8 pt]
    \item[(a)]  If $r \notin (a_i,b_i)$, $g_i(r)=0$  and
    \begin{itemize}[itemsep = 5 pt, topsep= 5 pt]
        \item \ $\bar{y}_j(i)  =  \frac{\bar{r} \big((1-\phi_i) \ell_i + (N-1) \phi_i + 2  \big) - (N+1) \mu_i }{(1-\phi_i)(N+1)(\ell_i+2)}$ \ \ for all $j \in \hat{N}_i$;
        \item \ $z_k = \frac{\bar{r}}{N+1}$ \ \ for all $k \in \mathbf{N} \setminus \left(\{i\} \cup  \hat{N}_i  \right)$;  
        \item $\bar{x}_i(r) =\frac{1}{2}\left(r - \sum_{j \in \hat{N_i}} \bar{y}_j(i) - (N - \ell_i -1)\frac{\bar{r}}{N+1}\right) $.  
    \end{itemize}  
    \item[(b)] If $r \in (a_i,b_i)$, $g_i(r)=1$  and
    \begin{itemize}[itemsep = 5 pt, topsep= 5 pt]
        \item \ $y_j(i,r)  = x_i(r) = \frac{1}{2+\ell_i}\left(r - (N - \ell_i -1)\frac{\bar{r}}{N+1}\right)$ \ \ for all $j \in \hat{N}_i$;
        \item $z_k = \frac{\bar{r}}{N+1}$ \  \ for all  $k \in \mathbf{N}\setminus \left(\hat{N}_i \cup \{i\} \right)$.
    \end{itemize}
    \item[(c)] Furthermore, the extreme points of the revelation interval $(a_i,b_i)$ are given by:
    \begin{itemize}[itemsep = 5 pt, topsep= 5 pt]
        \item $a_i  =  \frac{4 \frac{ N-\ell_i-1}{N+1} \bar{r} + \ell_i \frac{\bar{r}-\mu_i}{1-\phi_i}}{\ell_i+4}$
        \item   $b_i  =  \frac{\bar{r}-\mu_i}{1-\phi_i}$.
    \end{itemize}
\end{enumerate}
\end{proposition}
\medskip

The previous result provides closed-form expressions indicating how all contingent trades depend on the revelation intervals $(a_i,b_i)$ for every trader $i \in \mathbf{N}$ and therefore on the profile $\vec{\phi}=\left(\phi_i \right)_{i \in \mathbf{N}}$ of corresponding revelation probabilities given by (\ref{rev_prob}). 
There are two features arising from those expressions that are worth highlighting. One is that every trader
$k \in N \setminus (\{i\}\cup \hat N_i)$ who does not observe $r$ (i.e.\ who is not reached by the revelation of the
directly informed trader) trades the same amount,
$
z_k=\frac{\bar r}{N+1}
$,
independently of her network position and, in particular, independently of the identity of the directly informed
trader $i$. This follows directly from Proposition~\ref{prop_characterization}, which pins down $z_k$ as a constant across all such traders.
Note that it is also the equilibrium quantity that they would play in a simpler game where there is no private information at all.

As a consequence, the equilibrium prices and allocations characterized above are invariant to an alternative
informational specification in which these traders do \emph{not} observe the identity of the directly informed agent, if any.
In other words, whether the set $N \setminus (\{i\}\cup \hat N_i)$ knows (or does not know) who is informed, or whether anyone is, is
irrelevant for the equilibrium outcome.
Since Proposition \ref{prop_characterization} shows that the equilibrium allocation and prices are invariant 
to whether traders in $N \setminus (\{i\} \cup \hat N_i)$ observe the identity of 
the directly informed agent, in what follows we suppress this piece of 
meta-information from the notation, without affecting any equilibrium outcome.
\medskip

Another interesting feature of Proposition \ref{prop_characterization} pertains to the revelation intervals $(a_i,b_i)$ of each trader $i \in \mathbf{N}$. On the one hand, we have that $b_i$, the upper extreme of that interval, is equal to $r=  \frac{\bar{r}-\mu_i}{1-\phi_i}$, which coincides with the expected value held by the out-neighbors of $i$ when they observe that this trader does not reveal the realized value of $r$. The fact that such value is the upper bound of the revelation interval means that ``no revelation" amounts to good news relative to what they expect they would have learned if $i$ had chosen to reveal $r$.  
To see this, consider the situation where trader $i$'s out-neighbors have the following information: they know  that Nature has informed trader $i$ but this agent has not revealed that information. Note that when trader $i$ learns of such a value of $r$ she becomes indifferent between revealing it or not because both options lead to the same aggregate trade of her out-neighbors. Thus the induced trades -- including those of $i$ and her out-neighbors -- are all the same in either case and they vary continuously around such value of $r$. This contrasts with what happens at the lower extreme of the revelation interval. At that point, even though at equilibrium the originally informed trader $i$ must of course be indifferent between revealing or not, the induced trades of both $i$ and her out-neighbors are affected by the revelation decision. If the value of $r=a_i$ is revealed, this is ``bad news" for $i$'s out-neighbors, who cut down their trades as compared to what they would have been if they had been kept uninformed. In fact, as we shall see in Section \ref{S-Info_eff}, it is precisely such a discontinuity that leads, in our model, to the conclusion that the market may fail to be informationally efficient.\medskip

Two interesting conclusions follow from Proposition \ref{prop_characterization} as straightforward corollaries. \medskip

\begin{corollary}\label{cor_expect}
Let $\mathbf{s^*}$ be the PBE of the trading game, as given in (\ref{PBE_strat}), and let $\left( a_i^*, b_i^*  \right)$ be the (nonempty) revelation interval inducing the revelation decision $g^*_i(\cdot)$ of any given 
trader $i\in \mathbf{N}$.  Then,
\[a_i^*  < \frac{\mu_i^*}{\phi_i^*} < \bar{r}  < b_i^* 
\hspace{2 pt} 
\]
where $\bar{r}=\mathbb{E}(\tilde{r})$ and $ \frac{\mu_i^*}{\phi_i^*} =  \mathbb{E} (\tilde{r}  | \mbox{$i$ reveals at PBE})$.
\end{corollary}
\medskip

Corollary \ref{cor_expect} indicates that the prior expected value of the random variable $\tilde{r}$ (i.e., its mean $\bar{r}$) is always contained in the revelation interval and lies above the mean conditional on revelation, $\frac{\mu_i}{\phi_i}$. This simply derives from combing the following three observations: the latter conditional mean obviously belongs to the revelation interval; the upper end of the revelation interval is equal to the conditional mean upon non-revelation (cf. Part (c) of Proposition \ref{prop_characterization}); the unconditional mean $\bar{r}$ is a convex combination of two aforementioned conditional means.\medskip

A possible interpretation of the previous corollary is that, in our model, the mere anticipation that the originally informed trader will \textit{subsequently} reveal the value of $r$ amounts, in expected terms, to receiving ``bad news." To understand this point, the following thought experiment may be useful. Suppose that Nature has informed $i$ of the value of $r$ but the traders in $\hat{N}_i$ are \textit{only} told that $i$ will choose to reveal it to them (i.e. they do not yet get to know the value of $r$). What  would be their expected value of $r$ at that point? It would of course be $\frac{\mu_i}{\phi_i}$, which the corollary indicates is lower than the unconditional expected value $\bar{r}$ they held at the beginning of the game. They become, in other words, more pessimistic than they were at the start of the game. As the other side of the coin, this also implies that if those traders were told instead that $i$ will decide \textit{not} to reveal the value of $r$, their beliefs should turn to being more optimistic. Indeed, this latter conclusion immediately follows from the fact that, as already noted, upon such non-revelation the updated beliefs on $r$ by the out-neighbors of $i$ coincide with $b_i$, the upper end of the revelation interval (which is higher than $\bar{r}$ and therefore positive). In a heuristic sense, this can be described by saying that, for the out-neighbors of $i$, ``no news is good news." This in turn readily implies -- see (\ref{equ_trade_noninfomed}) in the Appendix -- the following corollary.

\medskip

\begin{corollary} \label{cor_pos_trade}
Let $\mathbf{s^*}$ be the PBE of the trading game, as given in (\ref{PBE_strat}). Then, for any $i \in \mathbf{N}$ and any of her out-neighbors $j \in \hat{N}_i$, if $i$ is informed of $r$ by Nature but does not reveal it, at the PBE the induced trade $\bar{y}_j^{\ast}(i) > 0$. 
\end{corollary}
\medskip

Along the lines of the previous discussion, Corollary \ref{cor_pos_trade} can be heuristically described as suggesting that when an originally informed agent $i$ decides not to reveal her information, this is interpreted by her outneighbors $j$ as ``good enough news" for their trading a positive amount at equilibrium. In fact, a similar
statement applies to all those other traders $k \in \mathbf{N}\setminus \left(\hat{N}_i \cup \{i\} \right)$ who do not observe the realization of $r$. For, in view of Part (b) of Proposition \ref{prop_characterization}, we know that these latter agents trade a fixed amount $z_k = \frac{\bar{r}}{N+1}$, which is also positive. So, in the end, we conclude that receiving no information on $r$, \textit{for whatever reason}, never dissuades traders from buying some positive amount of the asset. This in turn also implies a positive market price for the traded asset at equilibrium. \medskip

The discussion of the previous two corollaries has revolved around the question of how information revelation, or the absence of it, shapes updated beliefs and the induced trades. Next we turn our attention toward exploring how this same issue  depends on the communication network. We start the discussion by focusing on the simple and transparent context where the network is complete and every trader conveys her information to all others.


\begin{example}
\label{ex_N-1=S}
Consider the situation where some given $i \in \mathbf{N}$ has been informed of the value of $r$ by Nature. Then, in view of Proposition \ref{prop_characterization}, we know that the PBE of the game prescribes the following:
\begin{eqnarray}
\bar{y}_j & = & \frac{1}{N+1} \cdot \frac{\bar{r}-\mu_i}{1-\phi_i} > 0 \hspace{2 pt} , \nonumber \\
a_i & = & \frac{N-1}{N+3} \cdot  \frac{\bar{r}-\mu_i}{1-\phi_i} 
 \hspace{2 pt} ,
\nonumber \\
b_i & = & \frac{\bar{r}-\mu_i}{1-\phi_i} \hspace{2 pt} . \nonumber
\end{eqnarray}
Note that as $N = \ell +1$  grows, the ratio $\frac{a_i}{b_i}$ also grows, the \textit{ex-ante} probability that any informed trader $i$ reveals her information eventually shrinking to zero. 
\end{example}

The above example is a stylized manifestation of the following intuition. As the number of out-neighbors of the originally informed trader grows, the gains obtained by the latter from revelation (as it induces moderation on the out-neighbors' trading behavior) are progressively offset by the losses resulting from the entailed information leakage. In practice, this amounts to a shrinkage of the range of conditions for which revelation is optimal (i.e. a shortening of the revelation interval at equilibrium). When the network is incomplete or/and irregular, such trade-off must still be at work but its precise analysis becomes quite complex. For, in this case, it depends on the particular distribution of the random variable $\tilde{r}$ and its interplay with the details of the (potentially intricate) architecture of the communication network. \medskip

The previous discussion raises three questions on how to model the communication network. One of them, of course, is how the network forms in the first place, possibly before actual trading occurs. As advanced, this will be addressed in Section \ref{sec:endogenous}. But even if, as we have done so far, the network is taken as exogenously given, { two} other related questions also merit attention: \\
{ (a) Can we allow for the possibility that an originally informed agent may decide to convey the information only to a \textit{strict} subset of her out-neighbors? If feasible, the informed agent may indeed welcome this possibility if she had ``too many" out-neighbors (who, upon revelation, would become equally well informed competitors).} \\  
(b) Can we extend the model to give the traders who have been informed of the value of $r$ only indirectly (i.e. not by Nature) the option of relaying that information to her out-neighbors? This would be a natural feature to consider if, in contrast with what we have implicitly supposed, information does not become obsolete very fast. \medskip

{ Concerning question (a), first we claim that, even though the equilibria of the induced game would typically be much more complex than in our case, any of these equilibria would have a revelation structure that is a natural extension of the form displayed by the equilibrium characterized in Proposition \ref{prop_characterization}. In particular, the revelation decision of any given trader $i$ who is originally informed must be determined by (i) a \textit{finite} collection of \textit{bounded} intervals, each of them associated to a corresponding number of outneighbors to whom $i$ chooses to reveal her information when $r$ falls in it, complemented by (ii) a pair of unbounded intervals of the form $(-\infty,\hat{a}_i)$ and $( \hat{b}_i,\infty)$ such that if $r$ falls in either of them trader $i$ does \textit{not reveal} any information. Furthermore, without loss of generality, we can suppose that all those intervals are disjoint and the closure of their union covers the whole real line. Then we argue that most of the main insights of our model continue to hold. By way of illustration, in the next subsection we consider how the generalization contemplated here affects our key conclusion that the market fails to be informationally efficient. There we shall explain why it still applies if originally informed traders can control the number of their neighbors who gain access to their information. That is, even in that case, the equilibrium price does not \textit{ex post} reveal the whole information (i.e. the exact value of $r)$ that is collectively held by all market participants. }

\medskip

Concerning question (b), it is useful to address it by an extension of the model where revelation opportunities arise sequentially, each of them indexed by some $s=0,1,2,...$. At $s=0$, Nature reveals the value of $r$ to a randomly selected trader, just as in our baseline model. For all other subsequent $s\geq 1$, the agents who were just informed in the previous round $s-1$ are given the option -- if there are several of them, one at a time in some pre-specified order (e.g. index-based) --  whether they want to reveal their information to their not yet-informed out-neighbors. The process is bound to end at some (endogenous) finite $\hat{s}$, after which no further revelation occurs. \medskip

In the induced multi-stage finite game the players who have to take a revelation decision at any given round all know $r$. Therefore, if we assume for simplicity that the equilibrium played is in pure strategies and the information conveyed includes the identity of the trader who originally learned the value of $r$ from Nature, all deciding agents at any given $s$ enjoy common knowledge of the situation and can perfectly predict the ensuing behavior. This, of course, applies as well to all agents deciding at any later or earlier round, and in particular to the originally informed trader who is the one who decides whether to launch or not the communication wave in the first place. In view of it, we can describe the situation as one where there is just one round of communication on an artificially constructed network where the links connect the original holder of the information to all those traders to whom this information would reach in the multistage framework if that originally informed trader were to reveal her information. The analysis can then proceed as in the baseline model, the induced conclusions therefore being formally analogous as those obtained so far but having a multistage interpretation. 


\section{Informational efficiency}\label{S-Info_eff}

In this section we address the issue that (in line with the terminology often used in the literature) we have labelled before as the informational (in)efficiency of the market. In our context, it amounts to answering the question of whether or not the equilibrium price reveals the value of $r$ learned by the single trader $i$ who is originally informed by Nature. To this end, we need to understand how that price depends on the ``signal" $r$ that $i$ receives. More specifically, we have to check if the induced mapping is injective. Assuming for simplicity that $H$ and $Q$ are fixed and commonly known and therefore only $V$ is random and ex-ante unknown, the conclusion is as stated by the following result.

\begin{proposition}
\label{prop_price}
Under the specified conditions, in the unique equilibrium of the game, once Nature has chosen the directly informed trader $i$, the realized price under revelation is strictly lower than the price we would observe under no revelation.
Moreover, there exist two realizations $r'$ and $r''$, with $r' < a_i <  r'' < b_i $ such that the price corresponding to any realization of $r \in (r',a_i)$ is equivalent to the price corresponding to one and only one realization of $r \in (a_i,r'')$. 
\end{proposition}


In light of the benchmark inspired by \cite{kyle1985continuous}  discussed in the Introduction, Proposition \ref{prop_price} leads to the important insight that, if one allows for the possibility that traders may reveal their information to out-neighbors when they find it optimal to do so, the informational efficiency of the market breaks down. Conditional on any decision on revelation (reveal or not), the equilibrium-induced relationship between $r$ and the price is strictly monotone. Hence the prevailing price unambiguously identifies the value of $r$. However, when the revelation decision is endogenous and therefore can be affected by $r$, such a monotonicity no longer holds and hence the informational efficiency of the market fails. For, in this case, even a small increase in $r$ that leads the originally informed agent to reveal her information can substantially (i.e. discontinuously) impact her out-neighbors behavior. They trade less aggressively, which in turn leads to a lower equilibrium price. Other traders, however, who can only observe the price cannot know whether this happens because  $r$ is within the revelation interval (i.e., in the interval $(a_i, r'')$ in Figure \ref{fig:pV}) or below it (i.e., in the interval $(r',a_i)$).

\medskip



\begin{figure}
\begin{center}
\includegraphics[height=6cm]{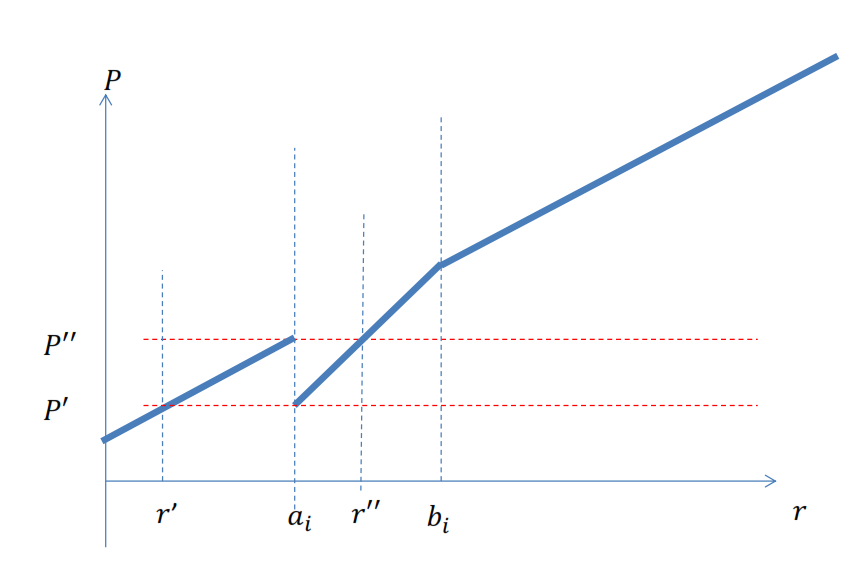}
\caption{Dependence of price from the realization of $r$.}
\label{fig:pV}
\end{center}
\end{figure}

As noted above, the informational inefficiency of the market extends to the generalized model where any originally informed agent $i$ can decide the number of out-neighbors to whom she reveals the realized value of $r$. To see this, it is enough to focus on the value $\hat{a}_i$ that marks the upper bound of the interval $(-\infty,\hat{a}_i)$ where $r$ is so low that $i$ chooses not to reveal it to any of her out-neighbors. Then, for some sufficiently small $\epsilon>0$, the equilibrium price for $r = \hat{a}_i+\epsilon$ must lie below that induced for $r = \hat{a}_i-\epsilon$, since the informed out-neighbors will reduce their trade in the former case (for they are becoming informed) as compared to the latter case (where they are not).\footnote{Note that here we rely on the assumption that any such revelation is a bilateral affair which is private to the two traders involved.} Thus, by a direct adaptation of the argument used for our baseline model, the observation of an equilibrium price in the interval $(P',P'')$ does not identify uniquely the realized value of $r$.\medskip

\section{Endogenous network}\label{sec:endogenous}

We now address the issue of how networks among traders form. Since it is natural to conceive inter-trader links as built on trust, we model network formation as a relatively stable social structure that is formed \textit{ex ante} before any particular information is realized. Also, we view trust as being bidirectional and therefore simplify the analysis by restricting to networks where every link is reciprocated (and therefore undirected).
\medskip

We rename players $\{1,2,...,N\}$ such that $p_1> p_2>...> p_N$. We call this the \textit{probability ordering}, which is common knowledge. For expositional purposes, we restrict the discussion to a context where there is strict ordering of such  probabilities but consider the possibility of ties in the Appendix. 
We say $i>j$ if $p_i>p_j$. \medskip

Some additional notation is required. Denote by $G$ an adjacency matrix that describes who can reveal information to whom. Let $G_i=(G_{i1},\ldots,G_{iN})\in\{0,1\}^{N}$ be the $i$th row of matrix $G$, which describes to whom $i$ can reveal information. We write $G_{ij}=1$ if $i$ can reveal information to $j$ and $0$ otherwise. By convention, we always impose $G_{ii}=0$.

Let $\mathcal{G}$ denote the set of \textit{symmetric matrices} in the above family, where $G_{ij}=G_{ji}$. Since links are assumed undirected, our network-formation model posits that they form bilaterally by bilateral consensus. 


We call a network $G$ \textit{ordered} if for any $i,j,k\in N$ with $p_i> p_j> p_k$, $G_{ik}=G_{ki}=1$ implies $G_{jk}=G_{kj}=1$ and $G_{ji}=G_{ij}=1$. 
Define a \textit{clique} $C(G)$ such that for any $i\in C(G)$, $G_{ij}=1$ if $j\in C(G)\setminus\{i\}$ and $G_{ij}=0$ if $j\not\in C(G)\setminus\{i\}$.\medskip

Trader $i$ knows everyone's probability of being directly informed and the distribution $f(r)$. 
As in the model with an exogenous network, whenever Nature chooses the directly informed trader $i$, the out-neighbors of $i$ know $i$ is directly informed. Proposition \ref{prop_existence} establishes equilibrium uniqueness of the trading game for any given network. Hence, strategic traders can infer all expected profits for any network. From an \textit{ex--ante} point of view, the network describes an allocation rule that assigns an expected profit to each strategic trader in the cases where she is the \textit{directly informed}, an \textit{out-neighbor} of the directly informed trader, or \textit{isolated}. The \textit{ex--ante} probabilities of players, together with the network, define the probabilities for each player to be in the respective position. 
We can thus interpret a network in the following way. 


\begin{definition}
    For any possible network $G\in\mathcal{G}$, define an allocation rule $\lambda:G\to \mathbb{R}^N$ that assigns each player $i$ an expected ex--ante payoff, $E(\pi_i)$.
\end{definition}

The allocation rule describes the expected payoffs of players for a particular network configuration. Since players infer their expected payoffs from the network, we only need to consider the linking strategies when considering stability of the network. Denote by $\sigma_i=(\sigma_{i1},\ldots,\sigma_{iN})\in\{0,1\}^{N}$ the \textit{linking strategy} of trader $i$, where $\sigma_{ij}=1$ indicates that $i$ want to establish a link to $j$ and $\sigma_{ij}=0$ indicates that she does not. We fix $\sigma_{ii}=0$ for all $i\in\mathbf{N}$. Let $\sigma=(\sigma_i,\sigma_{-i})$ be the strategy profile. Then, $G_{ij}=G_{ji}=\min\{\sigma_{ij},\sigma_{ji}\}$. In this way we define a function from every strategy profile $\sigma$ to an undirected network $G(\sigma) \in \mathcal{G}$. We employ the notion of \textit{bilateral equilibrium} (BE) \citep{goyal2007structural} to distinguish stable strategy profiles according to this \emph{network formation} process. 


\begin{definition}\label{def:bilateral}
    A strategy profile $\sigma^\ast$ is a bilateral equilibrium if
    \begin{enumerate}
        \item for any $i\in\mathbf{N}$ and every $\sigma_i\in \Sigma_i$, $E(\pi_i(\sigma^\ast))\geq E(\pi_i(\sigma_i,\sigma^\ast_{-i}))$, and
        \item for any pair $i,j\in\mathbf{N}$ and every $(\sigma_i,\sigma_j)$, $E(\pi_i(\sigma_i,\sigma_j,\sigma^\ast_{-i-j}))>E(\pi_i(\sigma^\ast_i,\sigma^\ast_j,\sigma^\ast_{-i-j}))$ implies $E(\pi_j(\sigma_i,\sigma_j,\sigma^\ast_{-i-j}))<E(\pi_j(\sigma^\ast_i,\sigma^\ast_j,\sigma^\ast_{-i-j}))$.
    \end{enumerate}
\end{definition} 

A BE of the network formation game establishes an allocation rule that is robust to joint deviations from any pair of players. No agent wants to alter her linking strategy, and there exists no pair of players who can simultaneously increase their profits by choosing any bilateral deviation. This is all in anticipation of the expected profit from the trading game. \medskip

Links are communication channels through which strategic traders can reveal information to others, thereby manipulating the market price in a favorable way. 
Whom a trader reveals her information to is irrelevant. However, links are communication channels through which strategic traders may receive information from their neighbors. Traders want to have neighbors with high probability of being directly informed as this determines their probability of being indirectly informed. 


\begin{proposition}\label{prop:endogenous}
    A BE of the network formation game exists. Moreover, any BE $\sigma$ profile induces the same network $G(\sigma) \in \mathcal{G}$, which is unique, ordered and consists of cliques.
\end{proposition}

Proposition \ref{prop:endogenous} establishes existence and uniqueness of the equilibrium network. Moreover, it provides a sharp characterization of the BE network. These results rest on the following observations. The probability for a player to indirectly learn the information increases in the probability that her neighborhood contains the directly informed trader. A directly informed trader is only concerned about how many traders she reveals the information to, but not their type. Higher types are consequently more desirable connections and link among themselves. The BE is ordered with respect to the \textit{ex-ante} probability to obtain information directly. The best type selects how many links she wants to establish and everyone would accept this link. \medskip

Trader $1$, who has the highest \textit{ex--ante} probability to be directly informed by Nature, can establish $\ell_1^\ast$ links. Any trader $2$ to $\ell_1^\ast+1$ accepts the links with $1$ and wants to establish exactly the same number of connections. This is because trader $1$ chooses the optimal number of out-neighbors subject to the distribution $f(r)$ and its effect on the revelation interval $(a_1,b_1)$. 
A similar logic applies to the remaining traders and the BE network consists of cliques. \medskip

Uniqueness of the BE network follows directly from the strict heterogeneity in traders' \textit{ex--ante} probabilities of being directly informed, which determine the exact composition of each clique. Since we chose an arbitrary (strict) probability ordering, this equilibrium must always exist. In principle, our results remain qualitatively similar in the general case where identical types are allowed. However, multiple BE networks may exist when identical traders are in different cliques. Swapping the position of those traders trivially yields an equilibrium, however, the different equilibria are economically identical. One example is the case where the worst clique includes players who cannot be directly informed at all but not all of these players. Exchanging traders with zero probability to be directly informed trivially yields equilibrium multiplicity. However, the different equilibria are strategically equivalent with identical clique sizes across equilibria. We provide more details on this in the Appendix. 

\section{Conclusion}\label{sec:conclusion}

This paper has explored the strategic information revelation within a trading model where traders operate under conditions of uncertainty about the ex-post value of a good. We considered a scenario where communication between traders is modeled as a network, and the true value of the good is known initially by only one trader, who then faces a critical decision: whether to share this information with her neighbors or to keep it concealed.\medskip

Our analysis demonstrates that in the unique equilibrium, strategic traders reveal only certain pieces of private information—specifically, those signals that lie within a critical interval centered around the mean of all possible signals. This behavior reflects a strategic consideration: by revealing information, a trader can subtly influence the beliefs and trades of their neighbors to their advantage, ultimately reducing their neighbors' demand and increasing their own profits. Good news is always concealed because it would stimulate demand and increase the price for the directly informed trader. Sufficiently bad news, which would lead to a greater drop in neighbors' trades and thus a substantial drop in the price, are also kept hidden to avoid disproportionately adverse effects on the profit. \medskip

The model also underscores the dual role of concealing information. Anticipating that slightly bad news are revealed, traders correct their belief about the good's value upward when information channels remain silent. 
This illustrates the complexity of strategic interactions in information-rich environments and the nuanced trade-offs that traders must navigate.\medskip

The selective revelation of information implies a positive probability for prices to be non-revealing the asset's true return. 
Since our model differs from the seminal framework of \cite{kyle1985continuous} by allowing for direct communication among traders and by assuming that the demand of the non-strategic sector is common knowledge, prices would be fully informative in the absence of communication.
Notably, the selective revelation of information remains robust to considering the possibility of communication channels forming strategically prior to Nature's choice of the directly informed trader. 

Our findings extend classical insights from oligopoly theory and revisit a classical question in the equilibrium theory of financial markets. While information is never revealed in models where firms compete in quantities, traders reveal certain signals despite competition in quantities (trades). This is because information revelation creates positive externalities for buyers and negative externalities for sellers.

	\appendix
	\global\long\def\thesection{Appendix \Alph{section}}
	\global\long\def\thesubsection{\Alph{section}.\arabic{subsection}}

\section{Proofs}

\begin{proof}[Proof of Lemma \ref {lemma_ab}]
When any agent $i$ is informed by Nature and she reveals this information to her out-neighbors, the aggregate quantity traded by all the (strategic) traders is:
\begin{eqnarray}
X & = & (1+\ell_i) x_i (r) + Z_i \nonumber \\
 & = & \frac{1 + \ell_i}{2 + \ell_i} r + \frac{1 }{2 + \ell_i} Z_i \ \ ,  \nonumber
\end{eqnarray}
and the common payoff for both the informed trader and those who have been informed indirectly is: 
\begin{eqnarray}
\pi_i & = & x_i (r) \frac{1}{K} \left( K V - X - H + Q \right)  \nonumber \\
 & = & \frac{1}{K} \frac{r-Z_i}{2 + \ell_i} \left( \frac{r}{2+\ell_i} -    \frac{Z_i}{2+\ell_i} \right) \nonumber \\
 & = & \frac{1}{K} \left( \frac{r-Z_i}{2 + \ell_i} \right)^2.
\label{eq:payoff_reveal}
\end{eqnarray}

\bigskip

Instead, if agent $i$ does not reveal her information, the aggregate quantity traded by all traders becomes:
\begin{eqnarray}
X & = & {\bar{x}}_i (r) + Z_i + \sum_{j \in \hat{N}_i } \bar{y}_j(i) \nonumber \\
 & = & \frac{1}{2} \left( r +  Z_i + \sum_{j \in \hat{N}_i } \bar{y}_j(i) \right)  \ \ ,  \nonumber
\end{eqnarray}
and the payoff for the only informed trader is 
\begin{eqnarray}
\pi_i & = & \bar{x}_i (r) \frac{1}{K} \left( K V - X - H + Q \right)  \nonumber \\
 & = & \frac{1}{2} \left(  r- Z_i - \sum_{j \in \hat{N}_i } \bar{y}_j(i) \right) \frac{1}{K} \left( r -   X  \right) \nonumber \\
 & = & \frac{1}{4K}\left(  r- Z_i - \sum_{j \in \hat{N}_i } \bar{y}_j(i) \right)^2 \ \ .
\label{eq:payoff_hide}
\end{eqnarray}

The payoff from equation (\ref{eq:payoff_reveal}) is greater or equal than the payoff from equation (\ref{eq:payoff_hide}) , when
\begin{equation}
\frac{1}{K}  \left(\frac{r-Z_i}{\left( 2 + \ell_i \right) }\right)^2 \geq \frac{1}{4K}\left(  r- Z_i - \sum_{j \in \hat{N}_i } \bar{y}_j(i) \right)^2 \ \ ,
\label{ineq:payoffs}
\end{equation}
or equivalently, when $r$ satisfies an inequality of the following form:
\begin{align*}
\underbrace{\left( \frac{1}{ \left( 2 + \ell_i \right)^2 } - \frac{1}{4} \right)}_{A} r^2 + \underbrace{\left(\frac{\sum_{j\in\hat{N}_i}\bar{y}_j(i)}{2}+\frac{Z_i}{2}-\frac{2Z_i}{(2+\ell_i)^2}\right)}_B r 
\\
- \underbrace{\left(\frac{\left(\sum_{j\in\hat{N}_i}\bar{y}_j(i)\right)^2}{4}+\frac{\left(\sum_{j\in\hat{N}_i}\bar{y}_j(i)\right)Z_i}{2}+\frac{Z_i^2}{4}-\frac{Z_i^2}{(2+\ell_i)^2}\right)}_C \geq 0.
\end{align*}
The left hand part is always a downward parabola (as $\ell_i\geq1$ is required for the revelation interval to be non-empty). Hence, if we find two values of $r_1$ and $r_2$ such that that $Ar_1^2+Br_1+C=0$ and $Ar_2^2+Br_2+C=0$, i.e., satisfy the expression with \textbf{equality}, then any $r$ that lies in--between $r_1$ and $r_2$ 
is revealed.
It is easy to check that both
\[
r_1 = \frac{2+\ell_i}{4+\ell_i} \left( \sum_{j \in \hat{N}_i } \bar{y}_j(i) \right) + Z_i
\]
and 
\[
r_2 = \frac{2+\ell_i}{\ell_i} \left( \sum_{j \in \hat{N}_i } \bar{y}_j(i) \right) + Z_i
\]
satisfy condition (\ref{ineq:payoffs}) with equality. Let us define $a_i=\min\{r_1,r_2\}$ and $b_i=\max\{r_1,r_2\}$. Lemma \ref{lemma_ab} follows.
\end{proof}

\bigskip


\begin{proof}[Proof of Lemma \ref{prop_existence}]
    As a first step, we establish that any $s$ that is a PBE of the game must be a separating equilibrium, where some signals are revealed (all $r\in(a_i,b_i)$) and all other signals are concealed. Lemma \ref{lemma_ab} immediately rules out an equilibrium where all signals are concealed. 

    { So, fix any $Z_i = \sum_{k\not\in \hat{N}_i\cup\{i\}}z_k \in \mathbb{R}$  and} suppose \textit{ad absurdum} that $g_i(r)=1$ for all $r\in \mathbb{R}$. 
    Any $k\not\in \hat{N}_i\cup\{i\}$ trades some quantity $z_k$, which directly informed trader $i$ can perfectly infer. The aggregate trade of isolated nodes is insensitive to the revelation decision of directly informed trader $i$, since traders who are not out-neighbors of $i$ do not observe $i$'s revelation decision by assumption. Let the posterior belief of an out-neighbor of $i$ off equilibrium, i.e., when trader $i$ chooses $g_i(r)=0$, be denoted by $\tilde{r}_j(i)$. The trading action of $j$ is thus given from equation \eqref{eq:opt_x} as:
    \begin{eqnarray}\label{eq:ybar}
        \bar{y}_j(i)=\frac{1}{2}\left(\tilde{r}_j(i)-\mathbb{E}_j(\bar{x}_i(r|r=\tilde{r}_j(i)))-\sum_{j'\in \hat{N}_i\setminus\{j\}}\mathbb{E}_j(\bar{y}_{j'}(\tilde{r}_{j'}(i)))-\sum_{k\not\in \hat{N}_i\cup\{i\}}z_k\right)
    \end{eqnarray}

    The other relevant posterior belief off equilibrium for the out-neighbor $i$ whom we call $j$ 
    is $\mathbb{E}_j(\bar{y}_{j'}(\tilde{r}_{j'}(i)))$, i.e., the beliefs about the trading action of each other out-neighbor of directly informed trader $i$, say $j'$ given the belief of $j$ about $\tilde{r}_{j'}$.
    
    Equation \eqref{eq:opt_x} uniquely pins down $\mathbb{E}_j(\bar{x}_i(r|r=\tilde{r}_j(i)))$ and $\mathbb{E}_j(\bar{y}_{j'}(\tilde{r}_{j'}(i)))$ for each $j,j'\in \hat{N}_i$. The quantity $\sum_{j\in \hat{N}_i}\bar{y}_j(i|r=\tilde{r}_j(i))$ is thus well defined for any off equilibrium beliefs the out-neighbors of $i$ may hold. Fix the off-path trades $\{\bar y_j(\tilde r_j(i))\}_{j\in\hat N_i}$ that $i$'s out-neighbors play upon observing concealment. Given any such fixed profile, Lemma~\ref{lemma_ab} implies that $i$'s best response is to reveal if and only if $r$ lies in the bounded interval $(a_i,b_i)$ defined in (9)--(10). Since $f$ has full support on $\mathbb R$, the event $\{r\notin(a_i,b_i)\}$ has strictly positive probability, and on it concealment strictly dominates revelation. Hence $g_i(r)=1$ for all $r$ cannot be optimal, contradicting the supposed pooling-on-revelation equilibrium. Any PBE is therefore separating. This occurs with positive probability since $f$ has full support on the real line. Hence, any PBE of the game is indeed a separating equilibrium. 

    We now show that the trading action of each $j\in \hat{N}_i$ (who may be informed indirectly) is unique in the PBE. Note, each $j\in \hat{N}_i$ can infer the trading action of $k\not\in \hat{N}_i\cup\{i\}$, so let us use the notation $Z_i\equiv\sum_{k\not\in \hat{N}_i\cup\{i\}}z_k$.

    The best response of $i$'s uninformed out-neighbor $j$ (i.e., when $g_i(r)=0$) is given by equation \eqref{eq:opt_x}: 

    \begin{eqnarray*}
        \bar{y}_j(i)=\frac{1}{2}\left(\underbrace{\frac{\bar{r}-\mu_i}{1-\phi_i}}_{\mathbb{E}_j(r|r\not\in(a_i,b_i)}-\mathbb{E}_j(\bar{x}_i(r))-(\ell_i-1)\bar{y}_j(i)-Z_i\right) ,
    \end{eqnarray*}
   { where $\phi_i \equiv \int_{a_i}^{b_i} dF(r) = F(b_i) - F(a_i)$ and $\mu_i \equiv \int_{a_i}^{b_i} r dF(r)$.}

The Bayesian posterior of $j$ about the trading action of informed trader $i$ follows from equation \eqref{eq:opt_x_hide}:
\[
\mathbb{E}_j (\bar{x}_i) = \frac{\bar{r} - \mu_i -  (1-\phi_i) \left( \ell_i \bar{y}_j(i) + Z_i \right)}{2 (1-\phi_i)}.
\]
The best response of uninformed out-neighbor $j$ is then
\begin{eqnarray}
\bar{y}_j(i) = \mathbb{E}_j (\bar{x}_i(r)) & = & \frac{1}{2(1-\phi_i)} \big( \bar{r} - \mu_i -  (1-\phi_i) \left( \ell_i \bar{y}_j + Z_i    \right) \big)  \nonumber \\
& = & \frac{ \bar{r} - \mu_i}{(\ell_i+2) (1-\phi_i)} - \frac{Z_i}{\ell_i+2}.
\label{eq:zs__zi}
\end{eqnarray}

By virtue of Lemma \ref{lemma_ab}, we arrive at the following equation:
%


\begin{eqnarray*}
   \underbrace{(2+\ell_i) \bar{y}_j(i)}_{LHS} &=&\mathbb{E}_j\left[\hspace{2 pt} \tilde{r}\  \vert  \ r \notin \bigg(\frac{2+\ell_i}{4+\ell_i}\hspace{2 pt}\ell_i\bar{y}_j(i)  +  Z_i\hspace{3 pt},\hspace{3 pt}(2+\ell_i)\bar{y}_j(i)+Z_i\bigg)\hspace{2 pt} \right] 
     -  Z_i 
 \nonumber \\
   \underbrace{(2+\ell_i) \bar{y}_j(i)+Z_i}_{LHS}  &=& \underbrace{\mathbb{E}_j\left[\hspace{2 pt} \tilde{r}\  \vert  \ r \notin \bigg(\frac{2+\ell_i}{4+\ell_i}\hspace{2 pt}\ell_i\bar{y}_j(i)+Z_i  \hspace{3 pt},\hspace{3 pt}(2+\ell_i)\bar{y}_j(i)+Z_i \bigg)\hspace{2 pt} \right]}_{RHS}
 \nonumber 
\end{eqnarray*}


%
which provides an implicit way of determining the value of $\bar{y}_j(i)$, given $Z_i$. Under the assumption that $\bar{y}_j(i)$ is non-negative we can argue that such a value exists and is well-defined using the following constructive argument.

Define $g(b)\equiv b-\Psi(b)$ on $(Z_i,\infty)$, where the interval $(a_i,b)$ is proper, with
$$
\Psi(b)=\dfrac{\int_{-\infty}^{a_i(b)} r\,dF(r)+\int_{b}^{\infty} r\,dF(r)}{F(a_i(b))+1-F(b)}$$
and
$a_i(b)=\alpha b+\tfrac{4Z_i}{\ell_i+4}$, $\alpha=\tfrac{\ell_i}{\ell_i+4}\in(0,1)$.

\emph{Existence.} As $b\downarrow Z_i$ the interval degenerates and $\Psi(b)\to\bar r$, so
$g(Z_i^+)=Z_i-\bar r<0$ (because {  we impose} $Z_i<\bar r$). As $b\to\infty$, $a_i(b)\to\infty$ with
$F(a_i(b))\to1$ and $1-F(b)\to0$, so $\Psi(b)\to\bar r$ and $g(b)\to\infty$. Since $g$ is
continuous, it has a zero $b_i^\ast\in(Z_i,\infty)$ by the intermediate value theorem.

\emph{Uniqueness.} Write $\Psi=M/D$ with $M(b)=\int_{-\infty}^{a_i}rdF+\int_b^\infty rdF$ and
$D(b)=F(a_i)+1-F(b)>0$. Then $M'(b)=\alpha\,a_i f(a_i)-b f(b)$ and
$D'(b)=\alpha f(a_i)-f(b)$. At any zero $b^\ast$ of $g$ we have $\Psi(b^\ast)=b^\ast$, so
\[
\Psi'(b^\ast)=\frac{M'-b^\ast D'}{D}=\frac{\alpha\,f(a_i^\ast)\,(a_i^\ast-b^\ast)}{D(b^\ast)}\le 0,
\]
since $a_i^\ast<b^\ast$. Hence $g'(b^\ast)=1-\Psi'(b^\ast)\ge 1>0$: $g$ is strictly increasing at
every zero, which rules out more than one zero. Therefore the fixed point $b_i^\ast$, and with it
$\bar y_j(i)=\tfrac{1}{2+\ell_i}(b_i^\ast-Z_i)$, is unique.

Intuitively, take \(b^\ast\) to be any fixed point of \(\Psi\), so that \(b^\ast=\Psi(b^\ast)\). A marginal rightward shift of the excluded interval removes probability mass whose average value is greater than the current truncated mean. Therefore, the truncated-conditional mean cannot increase locally, implying \(\Psi'(b^\ast)\leq 0\). Thus, at every fixed point, \(\Psi\) crosses the \(45^\circ\) line from above. This single-crossing property is sufficient for uniqueness. Moreover, the argument uses only the continuity of \(F\) and the existence of a density.



To conclude the proof, we verify that $\phi_i>0$, which implies a strictly positive
probability of revelation when $i$ has out-neighbors. Since $f$ has full support, $\phi_i>0$
holds whenever $b_i>a_i$. The proof proceeds by contradiction.

Suppose \textit{ad absurdum} that $a_i=b_i$ and $\hat N_i\neq\emptyset$. By
Lemma~\ref{lemma_ab} this holds whenever
\[
\frac{2+\ell_i}{4+\ell_i}\ell_i\,\bar y_j(i)+Z_i=(2+\ell_i)\bar y_j(i)+Z_i,
\]
and since in any PBE $\bar y_j(i)$ is unique given $Z_i$, as established previously,
$a_i=b_i$ implies $\bar y_j(i)=0$.

When $\phi_i=0$ revelation never occurs, so every out-neighbor $j\in\hat N_i$ keeps the prior
$\bar r$. Evaluating the out-neighbor best response \eqref{eq:ybar} at $\phi_i=\mu_i=0$ (the
limit as the revelation interval collapses) gives
\[
\bar y_j(i)=\frac{\bar r-\mu_i}{(\ell_i+2)(1-\phi_i)}-\frac{Z_i}{\ell_i+2}\bigg|_{\phi_i=\mu_i=0}
=\frac{\bar r-Z_i}{\ell_i+2}>0,
\]
since {  $Z_i
<\bar r$}. This contradicts $\bar y_j(i)=0$. Hence
$a_i\neq b_i$ and, by full support of $f$, $\phi_i=F(b_i)-F(a_i)>0$. \end{proof}







\bigskip

\begin{proof}[Proof of Proposition \ref{prop_characterization}]
{  Lemmas \ref{lemma_ab} and \ref{prop_existence} establish that any PBE has a separating structure, whereby the directly informed trader $i$ reveals any $r\in(a_i,b_i)$ and conceals any $r\notin(a_i,b_i)$, provided that $Z_i<\bar r$. We now close the equilibrium loop by showing that, for any candidate value of $\bar{y}_j(i)$, the aggregate trade $Z_i$ of agents who neither observe $r$ nor receive it through revelation is uniquely determined and independent of $\bar{y}_j(i)$. In particular, we show that in any PBE
$
Z_i=(N-\ell_i-1)\frac{\bar r}{N+1}<\bar r.
$
Therefore, the condition required in Lemma \ref{prop_existence} is always satisfied in equilibrium. This allows us to characterize the unique PBE.
}

    The trades in { any} PBE depend on $i$'s revelation decision, $g_i(r)$, which in turn depends on the expected trades of everyone else. Note, directly informed trader $i$ perfectly anticipates the trading action of $k\not\in \hat{N}_i\cup\{i\}$. 
    Trader $j\in\hat{N}_i$ trades precisely the same quantity as $i$ when $g_i(r)=1$, denoted by $x_i(r)$, because $i$ and $j$ hold the same information in this case. When $g_i(r)=0$, the equilibrium condition for $j$'s trading action is thus given from equation \eqref{eq:opt_x} as:

    \begin{eqnarray*}
        \bar y_j(i)=\frac12\left(\tilde r_j(i)-\bar x_i\big(\tilde r_j(i)\big)
-\sum_{j'\in\hat N_i\setminus\{j\}}\bar y_{j'}(i)-Z_i\right),\\ \qquad \text{with}
\qquad
\tilde r_j(i)\equiv \mathbb E_j\!\big(r\mid g_i(r)=0\big),\quad
Z_i\equiv\!\!\sum_{k\notin\hat N_i\cup\{i\}}\!\! z_k .
    \end{eqnarray*}

    The term $\tilde r_j(i)=\mathbb E_j(r\mid g_i(r)=0)$ captures trader $j$'s posterior belief about the realization of $r$ when concealment is observed. From equation \eqref{eq:opt_x}, we know that $\bar y_j(i)$ also depends on (i) the directly informed trader's trade evaluated at this posterior belief, $\bar x_i\big(\tilde r_j(i)\big)$; (ii) the total trade of all other out-neighbors of $i$, $\sum_{j'\in\hat N_i\setminus\{j\}}\bar y_{j'}(i)$; and (iii) the aggregate trade of traders outside $\hat N_i\cup\{i\}$, summarized by $Z_i=\sum_{k\notin\hat N_i\cup\{i\}} z_k$. Since trader $j$ can infer each $z_k$ for $k\notin\hat N_i\cup\{i\}$, the latter term enters directly rather than through an expectations operator.
    

    Since all $j,j'\in \hat{N}_i$ hold the same information, their trades in the PBE of the game must also be equal. Formally, it must hold in the PBE that $\bar{y}_j(i)=\bar{y}_{j'}(i)$. Using this fact, we can rewrite the condition as


    \[
    \bar{y}_j(i)=\frac{1}{2}\left(\mathbb{E}_j(r|r\not\in(a_i,b_i))-\bar{x}_i(r|r=\mathbb{E}_j(r|r\not\in(a_i,b_i)))-\sum_{j'\in \hat{N}_i\setminus\{j\}}\bar{y}_{j'}(i)-Z_i\right),
    \]

    where we use $Z_i\equiv \sum_{k\not\in\hat{N}_i\cup\{i\}}z_k$. Simplifying yields

    \[
    \bar{y}_j(i)=\frac{1}{\ell_i+1}\left(\mathbb{E}_j(r|r\not\in(a_i,b_i))-\bar{x}_i(r|r=\mathbb{E}_j(r|r\not\in(a_i,b_i)))-Z_i\right).
    \]

   Observing the network structure, trader $j$ can infer the revelation interval $(a_i,b_i)$, which also pins down the values $\phi_i$ and $\mu_i$. The posterior belief of any $j\in \hat{N}_i$ about $r$ is then
    
    \[
    \mathbb{E}_j (r  | r\not\in(a_i,b_i)) = \frac{\bar{r}- \mu_i}{1- \phi_i}.
    \]
    
    Conditional on $r\not\in (a_i,b_i)$, the posterior beliefs of $i$'s representative out-neighbor $j$ about $\bar{x}_i(r)$ is computed applying equation \eqref{eq:opt_x_hide}:

\[
\mathbb{E}_j (\bar{x}_i(r) | r\not\in(a_i,b_i)) = \frac{\bar{r} - \mu_i -  (1-\phi_i) \left( \ell_i \bar{y}_j(i) + Z_i \right)}{2 (1-\phi_i)}.
\]

The best response of uninformed out-neighbor $j$ can thus be rewritten as 
\begin{eqnarray}
\bar{y}_j(i) = \mathbb{E}_j (\bar{x}_i(r) | r\not\in(a_i,b_i)) & = & \frac{1}{2(1-\phi_i)} \big( \bar{r} - \mu_i -  (1-\phi_i) \left( \ell_i \bar{y}_j(i) + Z_i    \right) \big)  \nonumber \\
& = & \frac{ \bar{r} - \mu_i}{(\ell_i+2) (1-\phi_i)} - \frac{Z_i}{\ell_i+2}.
\label{eq:zs__zi}
\end{eqnarray}

This uses the fact that $\bar{y}_j(i)=\bar{y}_{j'}(i)$ for $j$ and $j'$ who are $i$'s out-neighbors, $j,j'\in\hat{N}_i$. Moreover, the expression uses the expected trading action of $i$ who chooses $g_i(r)=0$, as well as some aggregate trading action of isolated traders $Z_i$, which player $j$ is perfectly able to infer.

Note that any $k\not\in \hat{N}_i\cup\{i\}$, i.e., any $k$ who is neither an out-neighbor of directly informed trader $i$ nor the directly informed trader $i$, receives no additional information in the information revelation stage of the game. 
Using Bayes rule, the posterior belief of each trader $k$ who cannot be indirectly informed 
about $r$ is $\mathbb{E}_k(r) = \bar{r}$. The best response of trader $k$ can be expressed as:

\begin{equation}
z_k = \sum_{i\not\in \check{N}_k}p_i\left(\frac{1}{2} \left(  \bar{r} - \mathbb{E}_k (x_i(r)) - \ell_i \cdot \mathbb{E}_k (y_j(i,r)) - Z_{-k}\right)
\right)+\left(1-\sum_{i\not\in N}p_i\right)\frac{\bar{r}}{N+1}
\label{eq:zi__zs}
\end{equation}
where $Z_{-k}=Z_i-z_k$. The first term of the expression is the best response to the case where anyone whom $k$ cannot receive information from has been informed by Nature, weighted by the respective ex-ante probabilities with which these events occur. 
The second term computes the best response trade for $k$ in case no one has been informed, weighted by the appropriate ex-ante probability of this event. Note, trader $k$ assigns zero probability to the event where some $k'\in \check{N}_k$ holds information, i.e., someone who can pass on information to $k$. Indeed, this is consistent with the information available to $k$.
It is best response for player $k$ to play $z_k=\frac{\bar{r}}{N+1}$ if she believes Nature did not inform any trader. Moreover, for each $i$ who is possibly informed, $z_k=\frac{\bar{r}}{N+1}$ is also best response to the Bayesian posteriors $E_k(x_i(r))=\frac{\bar{r}}{N+1}$, $E_k(y_j(r,i))=\frac{\bar{r}}{N+1}$, $E_k(\bar{y}_j(i))=\frac{\bar{r}}{N+1}$ and $\mathbb{E}_k(z_{k'}|k'\not\in \hat{N}_i\cup\{i\})=\frac{\bar{r}}{N+1}$. Hence, $z_k=\frac{\bar{r}}{N+1}$ is the unique best response of any $k\not\in\hat{N}_i\cup\{i\}$. 
{  Therefore, among all the subgames considered in the statement of Lemma \ref{prop_existence}, only one can arise in equilibrium: the one in which
$
z_k=\frac{\bar r}{N+1}
$
for every trader $k$ who is neither informed nor can become informed through revelation. Consequently,
\[
Z_i=(N-\ell_i-1)\frac{\bar r}{N+1}<\bar r,
\]
so the condition required in Lemma \ref{prop_existence} is satisfied. This uniquely pins down the revelation interval and, therefore, the unique PBE of the whole game.
}

\medskip

{ To proceed with the characterization of this unique PBE}, we can use $Z_i=(N-\ell_i-1)\frac{\bar{r}}{N+1}$ when computing the best response $\bar{y}_j(i)$. For the case where $g_i(r)=0$, using $z_k=\frac{\bar{r}}{N+1}$, we obtain:

\begin{eqnarray*}
    \bar{y}_j(i) = \mathbb{E}_j (\bar{x}_i(r) | r\not\in(a_i,b_i)) & = & \frac{1}{2(1-\phi_i)} \big( \bar{r} - \mu_i -  (1-\phi_i) \left( \ell_i \bar{y}_j(i) + (N-\ell_i-1)z_k   \right) \big)  \nonumber \\
& = & \frac{ \bar{r} - \mu_i}{(\ell_i+2) (1-\phi_i)} - \frac{N-\ell_i-1}{\ell_i+2}\underbrace{\frac{\bar{r}}{N+1}}_{z_k}
\end{eqnarray*}

This is equivalent to the expression for $\bar{y}_j(i)$ in the statement of Proposition \ref{prop_characterization}. We can then write the equilibrium trade of directly informed trader $i$ as:

\begin{eqnarray*}
    \bar{x}_i(r)=\frac{1}{2}\left(r-\sum_{j\in \hat{N}_i}\bar{y}_j(i)-(N-\ell_i-1)\frac{\bar{r}}{N+1}\right)
\end{eqnarray*}

This is equivalent to the expression for $\bar{x}_i(r)$ in the statement of Proposition \ref{prop_characterization}. This concludes the proof of statement (a).

Having established the equilibrium trading action of $k\not\in \hat{N}_i\cup \{i\}$ and $j\in\hat{N}_i$, we can proceed with the trading action $x_i(r)$ when $i$ reveals her information. Given $g_i(r)=1$, trader $i$ and trader $j$ must choose the same trading action as they hold the same information. Hence:

\[
x_i(r)=y_j(i,r)=\frac{1}{2}\left(r-\ell_ix_i(r)-(N-\ell_i-1)\underbrace{\frac{\bar{r}}{N+1}}_{z_k}\right)
\]

Rewriting this expression yields:

\[
x_i(r)=y_j(i,r)=\frac{1}{\ell_i+2}\left(r-(N-\ell_i-1)\underbrace{\frac{\bar{r}}{N+1}}_{z_k}\right)
\]

This is indeed equivalent to the expression for $x_i(r)=y_j(i,r)$ in the statement of Proposition \ref{prop_characterization}. This concludes the proof of statement (b) because $g_i(r)=1$ iff $r\in(a_i,b_i)$ follows directly from Lemma \ref{lemma_ab}.

Having established all possible equilibrium trades, we can use the expressions from (a) and (b) to find the explicit values for $a_i$ and $b_i$. Since $\bar{r}\geq0$ by assumption, we can use the expression for $z_k$ to write:

\begin{eqnarray}\label{equ_trade_noninfomed}
    \bar{y}_j(i)= \frac{E_j(r|r\not\in(a_i,b_i))}{\ell_i+2}-\frac{N-\ell_i-1}{N+1}\frac{\bar{r}}{\ell_i+2}
\end{eqnarray}

Since $E_j(r|r\not\in(a_i,b_i))\geq\bar{r}$, we have that $\bar{y}_j(i)\geq0$ and thus:

\begin{eqnarray*}
    a_i & = & \frac{2+\ell_i}{4+\ell_i}\ell_i\bar{y}_i(i)+(N-\ell_i-1)\underbrace{\frac{\bar{r}}{N+1}}_{z_k}=\frac{4 \frac{ N-\ell_i-1}{N+1} \bar{r} + \ell_i \frac{\bar{r}-\mu_i}{1-\phi_i}}{\ell_i+4}\\
    b_i & = & \frac{2+\ell_i}{\ell_i}\ell_i\bar{y}_i(i)+(N-\ell_i-1)\underbrace{\frac{\bar{r}}{N+1}}_{z_k}=\frac{\bar{r}-\mu_i}{1-\phi_i}
\end{eqnarray*}
Indeed, $a_i<b_i$, which is consistent with Lemma \ref{lemma_ab}. This concludes the proof of Proposition \ref{prop_characterization}.
\end{proof}

\bigskip

\begin{proof}[Proof of Corollary \ref{cor_expect}]
    Suppose \textit{ad absurdum} that $a_i>b_i$. This would imply $\frac{\mu_i}{\phi_i}<a_i$. Proposition \ref{prop_characterization} establishes that $b_i=\frac{\bar{r}-\mu_i}{1-\phi_i}$ and $a_i>\bar{r}$ follows because, by assumption, $a_i>b_i$ (out-neighbor's posterior in case of no revelation) and $a_i>\frac{\mu_i}{\phi_i}$ (out-neighbor's posterior in case of revelation). 
 
    Now let us look at the expression of $a_i$. We have: 
    \[
    a_i=\frac{\ell_i}{\ell_i+4}b_i+\frac{4(N-\ell_i-1)}{(N+1)(\ell_i+4)}\bar{r}.
    \]
    
    Since the terms in front of $b_i$ and $\bar{r}$ are less than $1$ and $b_i<\bar{r}$ by assumption, in fact, $a_i<\bar{r}$. Indeed, if we suppose the extreme case where $b_i=\bar{r}$, we obtain $a_i=\frac{\ell_i}{\ell_i+4}\bar{r} + \bar{r}\frac{4}{\ell_i+4}\frac{N-\ell_i-1}{N+1}$. Since $\frac{N-\ell_i-1}{N+1}<1$ and $a_i<\bar{r}$, this is a contradiction because the unconditional mean of $r$ cannot be above both the conditional means in case of revelation and in case of no revelation. Corollary \ref{cor_expect} follows.
\end{proof}
\bigskip




\begin{proof}[Proof of Corollary \ref{cor_pos_trade}]
 {Take $j\in\hat{N}_i$ and $k\not\in\hat{N}_i\cup\{i\}$. By Proposition \ref{prop_characterization}, $z_k=E_k(\bar{x}_i|r=\bar{r})=\frac{\bar{r}}{N+1}>0$ and $\mathbb{E}_j(r|r\not\in(a_i,b_i))=b_i=\frac{\bar{r}-\mu_i}{1-\phi_i}>\bar{r}$. Since $\bar{y}_j(i)=E_j(\bar{x}_i|r=\frac{\bar{r}-\mu_i}{1-\phi_i})$, $\bar{y}_j(i)=E_j(\bar{x}_i|r=\frac{\bar{r}-\mu_i}{1-\phi_i})<z_k=E_k(\bar{x}_i|r=\bar{r})$ is inconsistent with Equation \ref{eq:opt_x} because:}

    \begin{align*}
        \bar{y}_j(i)&=&\frac{1}{2}\left(\underbrace{\frac{\bar{r}-\mu_i}{1-\phi_i}}_{>0}-\underbrace{(N-\ell_i-1)\frac{\bar{r}}{N+1}}_{>0}-(\ell_i)E_j(\bar{x}_i|r\not\in(a_i,b_i))\right)\\
        z_k & = & \frac12\left(\bar{r}-(N-\ell_i-2)\frac{\bar{r}}{N+1}-(\ell_i+1)\frac{\bar{r}}{N+1}\right)
    \end{align*}

    Here, we use the fact that $k$ expects everyone to have posterior belief $\bar{r}$. Rearranging yields:

    \begin{align*}
        (\ell_i+2)\bar{y}_j(i)&=&\underbrace{\frac{\bar{r}-\mu_i}{1-\phi_i}}_{>\bar{r}}-(N-\ell_i-1)\frac{\bar{r}}{N+1},\\
        (\ell_i+2)\frac{\bar{r}}{N+1} & = & \bar{r}-((N-\ell_i-1)\frac{\bar{r}}{N+1}.
    \end{align*}  
Hence $\bar y_j(i)>z_k>0$, which proves the claim.
\end{proof}

\bigskip

\begin{proof}[Proof of Proposition \ref{prop_price}]
First of all, when $H$ and $Q$ are fixed, $r$ is linear in $V$.
From equation (\ref{eq:price}) we have that 
\[
P = V + \frac{X-r}{K} \ \ .
\]
From equation \eqref{eq:opt_x_hide}, we have that outside the revelation interval
\begin{eqnarray}
P & = & V + \frac{1}{2K} \left( \ell_i \bar{y}_j(i)  + (N-\ell_i-1) z_k -  r \right) \nonumber \\
& = & 
\frac{1}{2} V + \frac{1}{K} \left( \ell_i \bar{y}_j(i)  + (N-\ell_i-1) z_k  +H - Q  \right)\ \ .
\label{eq:p_hide}
\end{eqnarray}

When instead $ r $ is in the revelation interval, the price resulting from equations (\ref{eq:price}) and  (\ref{eq:opt_x_reveal}), is such that
\begin{eqnarray}
P & = & V + \frac{N-\ell_i-1}{K} z_k - \frac{r}{K (\ell_i+2)} \nonumber \\
& = &  \frac{\ell_i+1}{\ell_i+2} V + \frac{N-\ell_i-1}{K} z_k + \frac{H-Q}{K (\ell_i+2)} \ \ .
\label{eq:p_reveal}
\end{eqnarray}

The two equations are linear in $V$, and  the steepness of $P$ with respect to $V$ is always higher in the interval of revelation than outside it, with this steepness increasing in the amount $\ell_i$ of possibly informed traders.

From Proposition \ref{prop_characterization} we also know that the revelation interval ends at 
\[
b_i= \frac{\bar{r}-\mu_i}{1-\phi_i} = \mathbb{E}_j (r  | r\not\in(a_i,b_i)) \ \ . 
\]
This means that at exactly $r=b_i$ out-neighbors do not receive new information and trade as they would without being informed, which also coincides with the trade of the directly informed trader. \\
So, for $r=b_i$ the two equations \eqref{eq:p_hide} and \eqref{eq:p_reveal} coincide.

\bigskip

Summing up, as $V$ (or equivalently $r$) increases, $P$ is linear with respect to $V$ for $r<a_i$, according to \eqref{eq:p_hide}.
At $r=a_i$ there is a discontinuity and the price, as $r$ increases, drops to the value predicted by the linearly increasing function \eqref{eq:p_reveal}. Then, the price $P$ is continuous and monotonically increasing in $r$. 
There is a value $r' <a_i$ such that $P(r') = \lim^+_{r \rightarrow a_i} P(r)$, and a value $r'' \in (a_i,b_i)$ such that $\lim^-_{r \rightarrow a_i} P(r) = P(r'')$.
\end{proof}

\bigskip

\begin{proof}[Proof of Proposition \ref{prop:endogenous}]

 To prove Proposition \ref{prop:endogenous} we start with a lemma establishing an important property of the bilateral equilibrium network.

    \begin{lemma}\label{lemma:ordered}
        Any network induced by a BE profile is ordered.
    \end{lemma}

    \begin{proof}
        Suppose \textit{ad absurdum} that for $i,j,k$ with $p_i> p_j> p_k\geq0$, $G_{ik}=1$ and $G_{ij}=0$. This implies that $\sigma^\ast_{ik}=1$ and either $\sigma^\ast_{ij}=0$ or $\sigma^\ast_{ji}=0$. We establish the argument for trader $i$. A similar argument holds for trader $j$. Denote by $\pi^d$, $\pi^o_m$ and $\pi^{iso}$ the expected profits of a player who is \textit{directly informed}, \textit{possibly indirectly informed} through trader $m$, or \textit{isolated}. Taking a trader's out-neighborhood as given, $\pi^d\geq\pi^o_m\geq\pi^{iso}$. 
        The expected \textit{ex--ante} payoff of generic trader $i$ is thus given by
        \begin{equation*}
            \mathbb{E}_i(\pi_i)=p_i\pi^d+\sum_{m\in N_i(G^\ast)}p_m\pi^o_m+\left(1-p_i-\sum_{m\in N_i(G^\ast)}p_m\right)\pi^{iso}
        \end{equation*}

        The first term is simply the \textit{ex--ante} payoff of being directly informed, weighted by the probability of this event. The second term is the \textit{ex--ante} payoff in case $i$ is possibly indirectly informed, weighted by the probability of this event. The third term is the \textit{ex--ante} payoff in case trader $i$ cannot be indirectly informed, weighted by the probability of this event.
        
        Consider first the case where $i$ is directly informed. The expression $\pi^d$ is independent of the identities of trader $i$'s out-neighbors. Hence, agent $i$ proposes connections with at least $\ell^\ast_i$ potential neighbors. 
        
        We can express the payoff of $i$ as $$\pi^d \;=\; \int_{a_i}^{b_i}\frac{1}{K}\left(\frac{r-Z_i}{\ell^*_i+2}\right)^{2}\,dF(r)
\;+\;\int_{r\notin(a_i,b_i)}\frac{1}{4K}\left(r-\ell^*_i\,\bar y_{s_i}(i)-Z_i\right)^{2}\,dF(r),$$


        Note that $a_i$, $b_i$, $\phi_i$, $\mu_i$, $Z_i$, and $\bar y_{s_i}(i)$ all depend on $i$'s out-neighbors only through their number $\ell_i^*$---never through their identities---since every excluded trader plays the same constant $z_k = \bar r/(N+1)$. Hence $\pi^d$ is a function of $\ell_i^*$ (and $f$) alone.
        
        Since $\pi^o_m\geq\pi^{iso}$, it is better to be linked to the directly informed trader than being isolated. Hence, player $i$ can establish weakly more links than $j$ and $k$, whose \textit{ex--ante} probability to receive the information is lower.
        
        Consider next the trade-off for an agent who is not directly informed. The second term of the \textit{ex-ante} payoff is the payoff from being a neighbor of the directly informed trader, weighted by the probability of a neighbor being directly informed.  The third term is the payoff of being isolated, weighted by the probability of being isolated. Together, these constitute the payoff of not being directly informed. The payoff of a neighbor of the directly informed trader exceeds the payoff of an isolated trader, i.e., $\pi^o_m\geq\pi^{iso}$. This is trivially true because an indirectly informed trader has more information compared to an isolated trader (even when the directly informed trader chooses to conceal her information). Player $i$ thus maximizes $\sum_{s_i\in N_i(G^\ast)}p_{s_i}$. Under the premise that $G^\ast_{ik}=1$ and $G^\ast_{ij}=0$, agents $i$ and $j$ do not obtain the highest possible outcomes, because she could delete the link to $k$ and establish a link to $j$ instead, which is also profitable for trader $j$. This contradiction proves the Lemma.
    \end{proof}
    
    Using Lemma \ref{lemma:ordered}, we can establish that the equilibrium network consists of cliques. We begin the argument by considering player $1$. Given $f(r)$, player $1$ establishes $\ell^\ast_1$ links. By Lemma \ref{lemma:ordered}, she forms connections to the best types, other than herself. 
    Hence, $G^\ast_{12}=...=G^\ast_{1\ell^\ast_1+1}=1$. We must then rule out that the trader indexed $\ell_1^\ast+1$ establishes additional links. Suppose \textit{ad absurdum} $G^\ast_{ij}=1$ for some $i\in \hat{N}_1(G^\ast)$ and $j\not\in \hat{N}_1(G^\ast)$. Lemma \ref{lemma:ordered} then implies $i$ has more neighbors than $1$, i.e., $\ell^\ast_1<\ell^\ast_i$. Trader $1$ and trader $i$'s expected draw in the revelation interval, subject to their out-degree, are $\frac{\mu_1}{\phi_1}<\bar{r}$ and $\frac{\mu_i}{\phi_i}<\bar{r}$. Their expected draw outside their revelation intervals is $b_1=\frac{\bar{r}-\mu_1}{1-\phi_1}$ and $b_i=\frac{\bar{r}-\mu_i}{1-\phi_i}$. Trader $1$ is the more desirable connection than $i$. Hence, given $G_{1i}=1$, trader $i$ can do no better than imitating trader $1$. Hence, $G_{1i}=1$  implies traders $\{1,\ldots,\ell_1^\ast+1\}$ indeed form a clique. A similar argument applies to all remaining players.

    The above procedure leads by construction to a BE, thereby also establishing existence of a BE network. 
\end{proof}

\section{Endogenous network}
\label{sec:endogenous-general}

In this section we allow for identical strategic traders, in terms their probability of being directly informed, when choosing their connections. 

Again, we rename players $\{1,2,...,N\}$ in a way such that $p_1\geq p_2\geq...\geq p_N$. In comparison to our main analysis, this nests the case $p_i=p_j$ for some $i,j\in \mathbb{N}$. We say $i\geq j$ if $p_i\geq p_j$. 


We introduce a notion of partial uniqueness of the BE network, $\mathbf{m}$-uniqueness. 

\begin{definition}\label{def:m-uniqueness}
    A BE network $\bar{\Gamma}^\ast$ is $\mathbf{m}$-unique if for $m=\{1,...,\mathbf{m}\}$, at least the $\mathbf{m}$ best cliques are identical in any other BE network $\bar{\Gamma}^{\ast\ast}$.
\end{definition}

The concept of $\mathbf{m}$-uniqueness captures the idea that we can characterize unique components of the equilibrium network. The intuition for why the bilateral equilibrium is unique for a strict probability ordering also translates to a partially strict probability ordering. Equilibrium multiplicity arises if a type, for whom an identical type exists is part of a clique whereas her twin is not. Trivially, replacing the two identical types also yields an equilibrium. The concept of $\mathbf{m}$-uniqueness thus characterizes the unique components of any BE network.

To characterize $\mathbf{m}$, take the highest type $i$ for whom there exists $j\in\mathbb{N}\setminus\{i\}$ with $p_i=p_j$. If no such type exists, uniqueness follows immediately from Proposition \ref{prop:endogenous}. We can characterize $\mathbf{m}$ as

\begin{equation}
    \mathbf{m}=\left\lfloor\frac{i-1}{S^\ast_1+1}\right\rfloor
\end{equation}

The index of the highest type for whom an identical player exists, $i$, is simply a number. Without losing generality, we can use the smallest index of players for whom an identical type exists. The value $S^\ast_1$ is the number of neighbors the best type establishes. The best clique is thus of size $S^\ast_1+1$. By Proposition \ref{prop:endogenous}, the best clique is also weakly larger than all other cliques. For a given value of $i$, $\left\lfloor \frac{i-1}{S^\ast_1+1}\right\rfloor$ is then the number of cliques sized $S^\ast_1+1$ we can form with $i-1$ players, which corresponds to the number of strictly ordered players and where $\lfloor\cdot\rfloor$ denotes the greatest integer less than or equal to. If $i-1<S^\ast_{1}+1$, there is potentially multiplicity in each component. 

Without fixing a specific distribution of the random variable $\tilde{r}$, we cannot pin down the optimal clique sizes and we rely on better cliques being weakly larger. The concept of $\bar{m}$-uniqueness thus establishes a lower bound to the unique components of the network.

\bibliographystyle{ecta}
\bibliography{biblio}

\end{document}